\documentclass[11pt]{amsart}

\usepackage{url}
\usepackage[colorlinks = true]{hyperref}
\usepackage{fullpage,amsmath,amsthm,amssymb,amsbsy}
\usepackage{paralist}
\usepackage{xcolor}
\usepackage{color}
\usepackage{graphicx}
\usepackage{algorithm,algpseudocode}
\usepackage{comment}

\usepackage{fancyhdr}
\usepackage{cite}
\usepackage{cleveref}

\newtheorem{thm}{Theorem}
\newtheorem{cor}{Corollary}
\newtheorem{prop}{Proposition}
\newtheorem{defi}{Definition}

\newtheorem{lem}{Lemma}
\theoremstyle{remark}
\newtheorem{remark}{Remark}
\theoremstyle{question}
\newtheorem{question}{Question}

\newcommand{\R}{\mathbb{R}}

\newcommand{\C}{\mathbb{C}}

\newcommand{\Z}{\mathbb{Z}}
\newcommand{\N}{\mathbb{N}}

\renewcommand{\P}[1]{\,\mathrm{P}#1}
\newcommand{\E}{\operatorname{E}}
\newcommand{\dif}{\operatorname{d}}

\newcommand{\vct}[1]{\boldsymbol{#1}}
\newcommand{\mtx}[1]{\boldsymbol{#1}}

\newcommand{\T}{\mathrm{T}}

\newcommand{\set}[1]{\mathbb{#1}}

\newcommand{\eps}{\epsilon}

\newcommand{\calI}{\mathcal{I}}
\newcommand{\calS}{\mathcal{S}}
\newcommand{\calT}{\mathcal{T}}

\newcommand{\calB}{\mathcal{B}}

\newcommand{\calN}{\mathcal{N}}
\newcommand{\calF}{\mathcal{F}}
\newcommand{\calA}{\mathcal{A}}

\newcommand{\calO}{\mathcal{O}}

\newcommand{\calX}{\mathcal{X}}

\newcommand{\calH}{\mathcal{H}}

\newcommand{\ve}{\vct{e}}

\newcommand{\vt}{\vct{t}}

\newcommand{\vx}{\vct{x}}
\newcommand{\vy}{\vct{y}}

\newcommand{\vlambda}{\vct{\lambda}}

\newcommand{\mB}{\mtx{B}}

\newcommand{\mF}{\mtx{F}}

\newcommand{\mH}{\mtx{H}}

\newcommand{\mQ}{\mtx{Q}}

\newcommand{\mS}{\mtx{S}}

\newcommand{\setA}{\set{A}}
\newcommand{\setB}{\set{B}}

\newcommand{\setG}{\set{G}}

\newcommand{\setM}{\set{M}}

\newcommand{\setS}{\set{S}}

\newcommand{\setU}{\set{U}}

\newcommand{\setW}{\set{W}}

\graphicspath{{./figs/}}

\newlength{\imgwidth}
\newcommand{\revise}[1]{\textcolor{black}{#1}}

\newboolean{twoColVersion}
\setboolean{twoColVersion}{false}
\newcommand{\twoCol}[2]{\ifthenelse{\boolean{twoColVersion}} {#1} {#2} }

\usepackage{enumitem}

\title[Time-Limited Toeplitz Operators]{Time-Limited Toeplitz Operators on Abelian Groups:
Applications in Information Theory and Subspace Approximation}

\author[Z. Zhu]{Zhihui Zhu}

\address[Z. Zhu]{Department of Electrical and Computer Engineering, University of Denver, Denver, CO 80208 USA}
\email{{\tt zhihui.zhu@du.edu}}

\author[M. Wakin]{Michael B. Wakin}
\address[M. Wakin]{Department of Electrical
Engineering, Colorado School of Mines, Golden, CO 80401 USA}
\email{\tt mwakin@mines.edu}



\begin{document}

\begin{abstract}
Toeplitz operators are fundamental and ubiquitous in signal processing and information theory as models for linear, time-invariant (LTI) systems. Due to the fact that any practical system can access only signals of finite duration, time-limited restrictions of Toeplitz operators are naturally of interest. To provide a unifying treatment of such systems working on different signal domains, we consider time-limited Toeplitz operators on locally compact abelian groups with the aid of the Fourier transform on these groups. In particular, we survey existing results concerning the relationship between the spectrum of a time-limited Toeplitz operator and the spectrum of the corresponding non-time-limited Toeplitz operator. We also develop new results specifically concerning the eigenvalues of time-frequency limiting operators on locally compact abelian groups. Applications of our unifying treatment are discussed in relation to channel capacity and in relation to representation and approximation of signals.
\end{abstract}

\maketitle


\section{Introduction}

This paper deals with generalizations of certain concepts from elementary signals and systems analysis, which we first review.

\subsection{Spectral analysis of linear, time-invariant systems}
\label{sec:LTIintro}

Linear, time-invariant (LTI) systems are ubiquitous in signal processing and control theory, and it is well known that the output of a continuous-time (CT) LTI system with input signal $x(t)$ can be computed using the convolution integral
\begin{align}\label{eq:convolution}
y(t) = (x \ast h) (t) = \int_{\tau=-\infty}^\infty h(t - \tau) x(\tau)\dif \tau,
\end{align}
where $h(t)\in L_2(\R)$ is the impulse response of the system. 
 Such a system can equivalently be viewed as a linear operator $\calH:L_2(\R)\rightarrow L_{\revise{\infty}}(\R)$,
where
\begin{align}\label{eq:CT toeplitz}
\calH(x)(t) = \int_{\tau=-\infty}^\infty h(t - \tau) x(\tau)\dif \tau.
\end{align}
Because this linear operator involves a kernel function $h(t-\tau)$ that depends only on the difference $t-\tau$, we refer to it as a {\em Toeplitz operator}.\footnote{Our notion of Toeplitz operators follows from the definition of Toeplitz operators in~\cite[Section 7.2]{grenander1958toeplitz}.} In this setting, the behavior of the Toeplitz operator can be naturally understood in the frequency domain: for an input signal $x(t)$ with continuous-time Fourier transform (CTFT)
\begin{align}
\widehat x(F) = \int_{t=-\infty}^\infty x(t)e^{-j2\pi Ft} \dif t, \forall \ F\in \R,
\label{eq:FTx}\end{align}
the CTFT of the output signal $y(t)$ will satisfy $\widehat y(F) = \widehat x(F)\widehat h(F)$, where $\widehat h(F)$ denotes the CTFT of the impulse response $h(t)$ and is also known as the {\em frequency response} of the system. The spectrum of the Toeplitz operator $\calH$ also coincides with $\widehat h(F)$ \cite[Section 7.2]{grenander1958toeplitz}. Note that the spectrum of a linear operator $\calH$, a generalization of the set of eigenvalues of a matrix, is the set of complex numbers $\lambda$ such that $\calH - \lambda \calI$ (where $\calI$ denotes the identity operator) is not invertible.

Similar facts hold for discrete-time (DT) LTI systems, where the response to an input signal $x[n]$ is given by the convolution
\begin{align}\label{eq:discrete convolution}
y[n] = \sum_{m=-\infty}^{\infty} h[n- m] x[m],
\end{align}
where $h[n]\in \ell_2(\Z)$ is the impulse response of the DT system. Such a system can equivalently be viewed as a linear operator $H:\ell_2(\Z)\rightarrow \ell_2(\Z)$, which corresponds to multiplication of the input signal $x \in \ell_2(\Z)$ by the bi-infinite Toeplitz matrix
\begin{align}
H=\left[\begin{array}{ccccc}\ddots & \ddots & \ddots & \ddots & \ddots \\
\ddots & h[0] & h[-1] & h[-2]  & \ddots \\
\ddots & h[1] & h[0] & h[-1] &  \ddots\\
\ddots & h[2] & h[1] & h[0] &  \ddots\\
\ddots & \ddots & \ddots & \ddots & \ddots
\end{array}\right].
\label{eq:DT toeplitz}\end{align}
We note that $H[m,n] = h[m-n] $ for all $m,n \in \Z$. The behavior of this system can also be interpreted as multiplication in the discrete-time Fourier transform (DTFT) domain where the DTFT of the impulse response $h[n]$ is defined as:
\begin{align}
\widehat h(f) = \sum_{n=-\infty}^\infty h[n]e^{-j2\pi fn}.
\label{eq:DTFT}
\end{align}
\revise{The spectrum of  $H$ also coincides with $\widehat h(f)$ \cite{grenander1958toeplitz}.}

\subsection{The effects of time-limiting}
\label{sec:effects of time-limiting}

Practical systems do not have access to input or output signals of infinite duration, which motivates the study of time-limited versions of LTI systems. Consider for example the situation where a CT system zeros out an input signal outside the interval $[0,T]$. (Or similarly, the system may pad with zeros an input signal that was originally compactly supported on $[0,T]$.) The system then computes the convolution shown in \eqref{eq:CT toeplitz} and finally time-limits the output signal to the same interval $[0,T]$. For such a situation we may define a new linear operator $\calH_T: L_2(\R) \rightarrow L_2(\R)$ (a ``time-limited'' version of $\calH$), where
\begin{align}\label{eq:CTTL}
\calH_T(x)(t) = \begin{cases} \int_{\tau=0}^{T} h(t - \tau) x(\tau) \dif \tau, & t \in [0,T] \\ 0, & \text{otherwise}. \end{cases}
\end{align}

An analogous time-limited version of $H$ (from~\eqref{eq:DT toeplitz}) may be defined for DT systems. Supposing that the input and output signals are time-limited to the index set $\{0,1,\dots,N-1\}$, we define the $N\times N$ Toeplitz matrix\footnote{Through the paper, finite-dimensional vectors and matrices are indicated by bold characters and we index such vectors and matrices beginning at $0$.} $\mH_N$ as
\begin{align}
\mH_N[m,n] = h[m-n], \quad \forall \ 0 \le m,n \le N-1.
\label{eq:FDT toeplitz}\end{align}
Such a matrix can also be viewed as a linear operator on $\C^N$.

A natural question is: {\em What effect do the time-limiting operations have on the system behavior?} More precisely, how similar is the spectrum of $\calH_T$ to that of $\calH$, and in what sense do the eigenvalues of $\calH_T$ converge to the frequency response $\widehat h(F)$ as $T \rightarrow \infty$?
Here $\calH_T$ is compact and thus has a discrete spectrum containing what we refer to as its eigenvalues; the number of eigenvalues is countable by the spectral theorem for compact operators\cite{conway2019course}.
Analogously, how similar is the spectrum of $\mH_N$ to that of $H$, and in what sense do the eigenvalues of $\mH_N$ converge to the frequency response $\widehat h(f)$ as $N \rightarrow \infty$? As we discuss, the answers to such questions provide insight into matters such as the capacity (or effective bandwidth) of time-limited communication channels and the number of degrees of freedom (or effective dimensionality) of certain related signal families. Answering these questions relies on deeper insight into the spectrum of Toeplitz operators.


\subsubsection{Toeplitz and time-limited Toeplitz operators}

In this paper, we distinguish between Toeplitz operators (such as $\calH$ and $H$) and time-limited Toeplitz operators such as $\calH_T$ and $\mH_N$.\footnote{Our usage of these terms is consistent with the terminology in~\cite[Section 7.2]{grenander1958toeplitz}, although in that work time-limited Toeplitz operators are referred to as {\em finite} Toeplitz operators.} We focus primarily on Toeplitz and time-limited Toeplitz operators that are Hermitian, i.e., $h(-t) = h^*(t)$ for $\calH$ and $\calH_T$ and $h[-n] = h^*[n]$ for $H$ and $\mH_N$.

We note that finite size Toeplitz {\em matrices} (such as $\mH_N$) are of considerable interest in statistical signal processing and information theory~\cite{grenander1958toeplitz,gray1972asymptotic,pearl1973coding,makhoul1975linear,kailath1978inverses}. The covariance matrix of a random vector obtained by sampling a wide-sense stationary (WSS) random process is an example of such a matrix. More general Toeplitz operators have been extensively studied since O.~Toeplitz and C.~Carath\'{e}odory~\cite{toeplitz1911theorie,caratheodory1911variabilitatsbereich}; see \cite{grenander1958toeplitz} for a very comprehensive review. Time-limited convolutions are also important in machine learning and computer vision. As an example, modern convolutional neural networks (CNNs)---which have demonstrated excellent performance in numerous computer vision tasks \cite{krizhevsky2012imagenet}---contain large numbers of convolutional layers, each of which mainly involves  two-dimensional convolution (that can be written as a doubly block circulant matrix, which is approximately Toeplitz) applied to the input.

Unfortunately, there are no simple formulas for the eigenvalues of time-limited Toeplitz operators such as $\calH_T$ and $\mH_N$. This stands in contrast to the operators $\calH$ and $H$, whose spectrum was given simply by the frequency response of the corresponding LTI system. Notably, although the discrete Fourier transform (DFT) is the canonical tool for frequency analysis in $\C^N$, the DFT basis vectors (complex exponentials of the form $e^{j2\pi \frac{nk}{N}}$ with $k \in \{0,1,\dots,N-1\}$) do not, in general, constitute eigenvectors of the matrix $\mH_N$, unless this matrix is {\em circulant} in addition to being Toeplitz. Consequently, the spectrum of $\mH_N$ cannot in general be obtained by taking the DFT of the time-limited impulse response $\{h[0],h[1],\dots,h[N-1]\}$.

Fortunately, in many applications it is possible to relate the eigenvalues of a time-limited Toeplitz operator to the eigenvalues of the original (non-time-limited) Toeplitz operator, thus guaranteeing that certain essential behavior of the original system is preserved in its time-limited version. We discuss these connections, as well as their applications, in more detail in the following subsections.

\subsubsection{Time-frequency limiting operators}
\label{sec:tfintro}

Shannon introduced the fundamental concept of capacity in the context of communication in \cite{shannon2001mathematical}, in which we find the answers to questions such as the capacity of a CT band-limited channel which operates substantially limited to a time interval $[0,T]$. In \cite{shannon2001mathematical}, the answer was derived in a probabilistic setting, while another notation of $\epsilon$-capacity was introduced by Kolmogorov in \cite{tikhomirov1993varepsilon} for approaching a similar question in  the deterministic setting of  signal (or functional) approximation. The functional approximation approach was further investigated by Landau, Pollak, and Slepian, who wrote a series of seminal papers exploring the degree to which a band-limited signal can be {\em approximately} time-limited~\cite{SlepiP_ProlateI,LandaP_ProlateII,LandaP_ProlateIII,Slepi_ProlateIV,Slepian78DPSS}. Recently, Lim and Franceschcetti~\cite{lim2014deterministic,lim2017information} provided a connection between Shannon's capacity from the probabilistic setting and Kolmogorov's capacity from the deterministic setting when communication occurs using band-limited functions.

To give a precise description, consider the case of a CT Toeplitz operator $\calH$ (as in~\eqref{eq:CT toeplitz}) that corresponds to an {\em ideal low-pass filter}. That is, $\calH =  \calB_W$, where $\calB_W: L_2(\R)\rightarrow L_2(\R)$ is a band-limiting operator that takes the CTFT of an input function on $L_2(\R)$, sets it to zero outside $[-W,W]$ and then computes the inverse CTFT. The impulse response of this system is given by the sinc function $h(t) = \frac{\sin(2\pi Wt)}{\pi W t}$, and the
frequency response of this system $\widehat h(F)$ is simply the indicator function of the interval $[-W,W]$.

Similarly, define $\calT_T: L_2(\R)\rightarrow L_2(\R)$ to be a time-limiting operator that sets a function to zero outside $[0,T]$, and finally consider the time-limited Toeplitz operator $\calH_T = \calT_T \calH \calT_T = \calT_T \calB_W \calT_T$. Observe that $\calH_T$ can be viewed as a composition of time- and band-limiting operators.

The eigenvalues of $\calT_T\calB_W\calT_T$ were extensively investigated in \cite{SlepiP_ProlateI,LandaP_ProlateII}, which discuss the ``lucky accident'' that $\calT_T\calB_W\calT_T$ commutes with a certain second-order differential operator whose eigenfunctions are a special class of functions---the {\em prolate spheroidal wave functions} (PSWFs).

The eigenvalues of the corresponding composition of time- and band-limiting operators in the discrete case, a Toeplitz matrix $\mH_N$ whose entries are samples of a digital sinc function, were studied by Slepian in~\cite{Slepian78DPSS}. The eigenvectors of this matrix are time-limited versions of the {\em discrete prolate spheroidal sequences} (DPSSs) which, as we discuss further in Section~\ref{sec:tfappreview}, provide a highly efficient basis for representing sampled band-limited signals and have proved to be useful in numerous signal processing applications.

In both the CT and DT settings, the eigenvalues of the time-limited Toeplitz operator exhibit a particular behavior: when sorted by magnitude, there is a cluster of eigenvalues close to (but not exceeding) $1$, followed by an abrupt transition, after which the remaining eigenvalues are close to $0$. This crudely resembles the rectangular shape of the frequency response of the original band-limiting operator. Moreover, the number of eigenvalues near $1$ is \revise{approximately equal to} the time-frequency area (which equals $2TW$ in the CT example above).  More details on these facts, including a complete treatment of the DT case, are provided in Section~\ref{sec:TF operator on group}.



\subsubsection{Szeg{\H o}'s theorem}\label{sec:conventional Szego}

For more general Toeplitz operators---beyond ideal low-pass filters---the eigenvalues of the corresponding time-limited Toeplitz operators can be described using Szeg{\H o}'s theorem.

We describe this in the DT case. Consider a DT Hermitian Toeplitz operator $H$ which corresponds to a system with frequency response $\widehat h(f)$, as described in~\eqref{eq:DT toeplitz} and~\eqref{eq:DTFT}. For $N \in \N$, let $\mH_N$ denote the resulting time-limited Hermitian Toeplitz operator, as in~\eqref{eq:FDT toeplitz}, and let the eigenvalues of $\mH_N$ be arranged as $\lambda_0(\mH_N)\geq \cdots\geq \lambda_{N-1}(\mH_N)$. Suppose $\widehat h\in L_{\infty}([0,1])$. Then Szeg\H{o}'s theorem~\cite{grenander1958toeplitz} describes the collective asymptotic behavior (as $N\rightarrow \infty$) 
of the eigenvalues of the sequence of Hermitian Toeplitz matrices $\{\mH_N\}$ as
\begin{equation}
\lim_{N\rightarrow \infty}\frac{1}{N}\sum_{l=0}^{N-1}\vartheta(\lambda_l(\mH_N)) = \int_0^{1}\vartheta(\widehat{h}(f))df,
\label{eq:Szego thm}\end{equation}
where $\vartheta$ is any function continuous on the range of $\widehat h$. As one example, choosing $\vartheta (x) = x$ yields
\begin{align*}
\lim_{N\rightarrow \infty}\frac{1}{N}\sum_{l=0}^{N-1}\lambda_l(\mH_N) = \int_0^{1}\widehat h(f)df.
\end{align*}
In words, this says that as $N\rightarrow \infty$, the average eigenvalue of $\mH_N$ converges to the average value of the frequency response $\widehat{h}(f)$ of the original Toeplitz operator $H$. As a second example, suppose $\widehat h(f)>0$ (and thus $\lambda_l(\mH_N)>0$ for all $l\in \{0,1,\dots,N-1\}$ and $N \in \N$) and let $\vartheta$ be the $\log$ function. Then Szeg\H{o}'s theorem indicates that
\begin{align*}
\lim_{N\rightarrow \infty}\frac{1}{N}\log\left(\det\left( \mH_N\right) \right) = \int_0^{1}\log\left(\widehat h(f)\right) df.
\end{align*}
This relates the determinants of the Toeplitz matrices $\mH_N$ to the frequency response $\widehat{h}(f)$ of the original Toeplitz operator $H$.

Szeg\H{o}'s theorem has been widely used in the areas of signal processing, communications, and information theory. A paper and review by Gray~\cite{gray1972asymptotic,gray2005toeplitz} serve as a remarkable elementary introduction in the engineering literature and offer a simplified proof of Szeg\H{o}'s original theorem. The result has also been extended in several ways. For example, the Avram-Parter theorem~\cite{avram1988bilinear,parter1986distribution} relates the collective asymptotic behavior of the singular values of a general (non-Hermitian) Toeplitz matrix to the magnitude response $|\widehat h(f)|$. Tyrtyshnikov~\cite{tyrtyshnikov1996unifying} proved that Szeg\H{o}'s theorem holds if $\widehat h(f)\in \R$ and $\widehat h(f)\in L^{2}([0,1])$, and Zamarashkin and Tyrtyshnikov~\cite{zamarashkin1997distribution} further extended Szeg\H{o}'s theorem to the case where $\widehat h(f)\in \R$ and $\widehat h(f)\in L^{1}([0,1])$. Sakrison~\cite{sakrison1969extension} extended Szeg\H{o}'s theorem to higher dimensions. Gazzah et al.~\cite{gazzah2001asymptotic} and Guti\'{e}rrez-Guti\'{e}rrez and Crespo~\cite{gutierrez2008asympBolckToeplitz} extended Gray's results on Toeplitz and circulant matrices to block Toeplitz and block circulant matrices and derived Szeg\H{o}'s theorem for block Toeplitz matrices.

Similar results also hold in the CT case, with the operators $\calH$ and $\calH_T$ as defined in~\eqref{eq:CT toeplitz} and~\eqref{eq:CTTL}. Let $\lambda_{\ell}(\calH_\T)$ denote the $\ell^{\text{th}}$-largest eigenvalue of $\calH_T$. Suppose $\widehat h(F)$ is a real-valued, bounded and integrable function, i.e., $\widehat h(F) \in \R$, $\widehat h(F)\in L^{\infty}(\R)$, and $\widehat h(F)\in L^{1}(\R)$. Then Szeg\H{o}'s theorem in the continuous case~\cite{grenander1958toeplitz} states that the eigenvalues of $\calH_T$ satisfy
\begin{align}
\lim_{T\rightarrow \infty}\frac{\#\{\ell:a< \lambda_{\ell}(\calH_{T})< b\}}{T} = \left|\{F:a<\widehat h (F)<b\}\right|
\label{eq:szego continuous}\end{align}
for any interval $(a,b)$ such that $|\{F:\widehat h(F) = a\}| = |\{F:\widehat h(F) = b\}| = 0$. Here $|\cdot|$ denotes the length (or Lebesgue measure) of an interval. Stated differently, this result implies that the eigenvalues of the operator $\calH_T$ have asymptotically the same distribution as the values of $\widehat h(F)$ when $F$ is distributed with uniform density along the real axis.

We remark that although the collective behavior of the eigenvalues of the time-frequency limiting operators discussed in Section~\ref{sec:tfintro} can be characterized using Szeg\H{o}'s theorem, finer results on the eigenvalues have been established for this special class of time-limited Toeplitz operators. We discuss Szeg\H{o}'s theorem for general operators in Section~\ref{sec:General Toeplitz Operators}, and we discuss results for time-frequency limiting operators in Section~\ref{sec:TF operator on group}.


\subsubsection{Time-limited Toeplitz operators on locally compact abelian groups}

One of the important pieces of progress in harmonic analysis made in last century is the definition of the Fourier transform on locally compact abelian groups~\cite{rudin2011fourier}. This framework for harmonic analysis on groups not only unifies the CTFT, \revise{$n$-dimensional CTFT}, DTFT, and DFT (for signal domains, or groups, corresponding to $\R$, \revise{$\R^n$}, $\Z$, and $\Z_N := \{0,1,\dots,N-1\}$, respectively), but it also allows these transforms to be generalized to other signal domains. This, in turn, makes possible the analysis of new applications such as steerable principal component analysis (PCA)~\cite{vonesch2015steerable} where the domain is the rotation angle on $[0,2\pi)$, an imaging system with a pupil of finite size~\cite{di1969degrees}, line-of-sight (LOS) communication systems with orbital
angular momentum (OAM)-based orthogonal multiplexing techniques~\cite{xu2017degrees}, and many other applications such as those involving rotations in three dimensions~\cite[Chapter 5]{chirikjian2016harmonic}.

In this paper, we consider the connections between Toeplitz and time-limited\footnote{\revise{Here we use ``time'' to be consistent with the preceding discussion, but this concept broadly applies to other domains such as the spatial domain.}} Toeplitz operators on locally compact abelian groups. As we review in Section~\ref{sec:preliminary}, one important fact carries over from the classical setting described in Section~\ref{sec:LTIintro}: the eigenvalues of any Toeplitz operator on a locally compact abelian group are given by the generalized frequency response of the system.

In light of this fact, we are once again interested in questions such as: How does the spectrum of a time-limited Toeplitz operator relate to the spectrum of the original (non-time-limited) Toeplitz operator? In what sense do the eigenvalues converge as the domain of time-limiting approaches the entire group? The answers to such questions will provide new insight into the effective dimensionality of certain signal families (such as the class of signals that are time-limited and essentially band-limited) and the amount of information that can be transmitted in space or time by band-limited functions.

\subsection{Contribution and paper organization}

\revise{This paper focuses on the spectra of time-limited Toeplitz operators and the resulting implications in signal processing and information theory, containing part survey and part novel work. In particular, as new results, we study the spectra of time-frequency limiting operators on locally compact abelian groups and analyze  applications in representation and approximation of band-limited signals, generalizing the existing results in \Cref{sec:tfintro}.
}

The remainder of paper 
is organized as follows. Section~\ref{sec:preliminary} reviews harmonic analysis on locally compact abelian groups and draws a connection between time-limited Toeplitz operators and the effective dimensionality of certain related signal families. Next, Section~\ref{sec:General Toeplitz Operators} reviews Szeg\H{o}'s theorem and its (existing) generalization to locally compact abelian groups. Applications are discussed in channel capacity, signal representation, and numerical analysis. Finally, Section~\ref{sec:TF operator on group} reviews existing results on the eigenvalues of time-frequency limiting operators and generalizes these results to locally compact abelian groups. New applications of this unifying treatment are discussed in relation to channel capacity and in relation to representation and approximation of signals.  This work also opens up new questions concerning applications and research directions, which we discuss at the ends of \Cref{sec:General Toeplitz Operators} and \Cref{sec:TF operator on group}.

\section{Preliminaries}\label{sec:preliminary}

We briefly introduce some notation used throughout the paper. Sets (of variables, functions, etc.)\ are denoted in blackboard font as $\setA, \setB,\ldots$. Operators are denoted in calligraphic font as $\calA, \calB,\ldots$.

\subsection{Harmonic analysis on locally compact abelian groups}

\subsubsection{Groups and dual groups}

Let $\setG$ (with a binary operation $\circ$) denote a  {\em locally compact abelian group}, 
which can be either discrete or continuous, and either compact or non-compact. A {\em character} $\chi_{\xi}:\setG\rightarrow \mathbb T$ of $\setG$ is a continuous group homomorphism from $\setG$ with values in the circle group $\mathbb T :=\left\{z\in\C:|z|=1\right\}$ satisfying
\begin{align*}
&\left|\chi_{\xi}(g)\right| = 1,\ \chi_\xi^*(g) = \chi_\xi(g^{-1}),\ \chi_\xi (h\circ g) = \chi_\xi (h)\chi_\xi (g),
\end{align*}
for any $g,h\in\set G$. Here $\chi_\xi^*(g)$ is the complex conjugate of $\chi_\xi(g)$.
The set of all characters on $\setG$ introduces a locally compact abelian group, called the {\em dual group} of $\setG$ and denoted by $\widehat \setG$ if we pair $(g,\xi)\rightarrow \chi_{\xi}(g)$ for all $\xi\in \widehat\setG$ and $g\in \setG$. In most references the character is denoted simply by $\chi$ rather than by $\chi_\xi$. However, we use here the notation $\chi_\xi$ in order to emphasize that the character can be viewed as a function of two elements $g\in \setG$ and $\xi \in \widehat \setG$, and for any $\xi\in\widehat \setG$, $\chi_\xi$ is a function of $g$. In this sense, $\chi_\xi(g)$ can be regarded as the value of the character $\chi_\xi$ evaluated at the group element $g$. Table~\ref{table:examples of FT} lists several examples of groups $\setG$, along with the corresponding binary operation $\circ$ and dual group $\widehat \setG$, that have relevance in signal processing and information theory. Here
$\text{mod}(a,b) = \frac{a}{b} - \lfloor \frac{a}{b} \rfloor$, where $\lfloor c \rfloor$ is the largest integer that is not greater than $c$. 

\begin{table}[h!]\caption{Examples of groups $\setG$, along with their dual groups $\widehat\setG$ and Fourier transforms (FT). Below, CT denotes continuous time, DT denotes discrete time, FS denotes Fourier series, and DFT denotes the discrete Fourier transform. }\label{table:examples of FT}
\begin{center}
\small
\begin{tabular}{c|c|c|c|c|c|c}
\hline group $\setG$ & dual group $\widehat\setG$& $g$& binary operation $g_1\circ g_2$ & $\xi$ & $\chi_\xi(g)$ & FT\\
\hline \hline $\R$ & $\R$  & $t$ & $t_1+t_2$  &$F$ & $e^{j2\pi Ft}$ & CTFT \\
\hline $\R^n$ & $\R^n$  & $\vt$ & $\vt_1 + \vt_2$  &$\mF$ & $e^{j2\pi \langle\mF, \vt\rangle}$ & CTFT in $\R^n$\\
\hline unit circle $[0,1)$& $ \Z   $  & $t$ & $\text{mod}(t_1 + t_2,1)$  &$ k$ & $e^{j2\pi tk}$ & CTFS \\
\hline $\Z$ & unit circle & $n$  & $n_1 + n_2$ & $f$ & $e^{j2\pi fn}$ & DTFT \\
\hline $\Z_{N} = N$ roots of unity& $\Z_{N} = N$ roots of unity  & $ n$ & $\text{mod}(n_1 +n_2,N)$  &$k$ & $e^{j2\pi\frac{nk}{N}}$ & DFT\\
\hline
\end{tabular}
\end{center}
\end{table}

\subsubsection{Fourier transforms}
\label{sec:FT groups}
The characters $\{\chi_\xi\}_{\xi\in\widehat\setG}$ play an important role in defining the Fourier transform for functions in $L_2(\setG)$. In particular, the Pontryagin duality theorem~\cite{rudin2011fourier}, named after Lev Semennovich Pontryagin who laid the foundation for the theory of locally compact abelian groups, generalizes the conventional CTFT on $L_2(\R)$ and CT Fourier series for periodic functions to functions defined on locally compact abelian groups.
\begin{thm}[Pontryagin duality theorem~\cite{rudin2011fourier}]
Let $\setG$ be a locally compact abelian group and $\mu$ be a Haar measure on $\setG$. Let $x(g)\in \revise{L_1}(\setG)$. Then the Fourier transform $\widehat x(\xi)$ is defined by $
\widehat x(\xi) = \int_\setG x(g)\chi_\xi^*(g) \dif \mu (g). $
For each Haar measure $\mu$ on $\setG$ there is a unique Haar measure $\nu$ on $\widehat \setG$ such that the following inverse Fourier transform holds $
x(g) = \int_{\widehat \setG}\widehat x(\xi)\chi_\xi(g) \dif \nu(\xi).$
The Fourier transform satisfies Parseval's theorem:
$
\int_\setG \left|x(g)\right|^2 \dif \mu(g) = \int_{\widehat \setG}\left|\widehat x(\xi)\right|^2 \dif \nu(\xi).$
\label{thm:FT}\end{thm}
With certain technical tricks, the same Fourier transform and Parseval's theorem in \Cref{thm:FT} can be extended to $L_2(\setG)$; see \cite[Theorem 59]{tao:FT}, \cite[1.6.1]{rudin2011fourier} and the references therein for details. Only Haar measures and integrals are considered throughout this paper. We note that the unique Haar measure $\nu$ on $\widehat \setG$ depends on the choice of Haar measure $\mu$ (\revise{which is unique except for positive scaling factors}) on $\setG$. We illustrate this point with the conventional DFT as an example where $g = n \in \setG =\Z_N$, $\xi = k\in \widehat\setG = \Z_N$, and $\chi_\xi(g) = e^{j2\pi \frac{nk}{N}}$. If we choose the counting measure (where each element of $\setG$ receives a value of 1) on $\setG$, then we must use the normalized counting measure (where each element of $\widehat\setG$ receives a value of $\frac{1}{N}$) on $\widehat \setG$. The DFT and inverse DFT become
\[
\widehat \vx[k] = \sum_{n=0}^{N-1}\vx[n]e^{-j2\pi \frac{nk}{N}};\ \vx[n] = \frac{1}{N}\sum_{k=0}^{N-1}\widehat\vx[k]e^{j2\pi \frac{nk}{N}}.
\]
One can also choose the semi-normalized counting measure (where each element receives a value of $\frac{1}{\sqrt{N}}$) on both groups $\setG$ and $\widehat\setG$. This gives the normalized DFT and inverse DFT:
\[
\widehat \vx[k] = \frac{1}{\sqrt{N}}\sum_{n=0}^{N-1}\vx[n]e^{-j2\pi \frac{nk}{N}};\ \vx[n] = \frac{1}{\sqrt{N}}\sum_{k=0}^{N-1}\widehat\vx[k]e^{j2\pi \frac{nk}{N}}.
\]

For convenience, we rewrite the Fourier transform and inverse Fourier transform as follows when the Haar measures are clear from context:
\[
\widehat x(\xi) = \int_\setG x(g)\chi_\xi^*(g) \dif g;\
x(g) = \int_{\widehat\setG}\widehat x(\xi)\chi_\xi(g) \dif \xi.
\]
We also use $\calF:L_2(\setG)\rightarrow L_2(\widehat \setG)$ and $\calF^{-1}: L_2(\widehat \setG)\rightarrow L_2(\setG)$ to denote the operators corresponding to the Fourier transform and inverse Fourier transform, respectively.

For each group $\setG$ and dual group $\widehat \setG$ listed in Table~\ref{table:examples of FT}, the table also includes the corresponding Fourier transform.

\subsubsection{Convolutions}

For any $x(g),y(g)\in L_2(\setG)$, we define the convolution between them by
\begin{align}
(x \star y)(g):= \int_\setG y(h) x(h^{-1}\circ g) \dif h.
\label{eq:def conv group}\end{align}
Similar to what holds in the standard CT and DT signal processing contexts, it is not difficult to show that the Fourier transform on $\setG$ also takes convolution to multiplication. That is,
\begin{align*}
&\calF(x \star y)(\xi) = \int_{\setG} \int_\setG y(h) x(h^{-1}\circ g) \dif h\;  \chi_\xi^*(g) \dif g\\&= \int_{\setG} \int_\setG  x(h^{-1}\circ g) \chi_\xi^*(h^{-1}\circ g) \dif g \; \chi_\xi^*(h) y(h) \dif h = (\calF x)(\xi) (\calF y)(\xi)
\end{align*}
since $\int_{\setG} x(h^{-1}\circ g) \dif g = \int_{\set G} x(g) \dif g$ for any  $h\in\setG$.

Similar to the fact that Toeplitz operators~\eqref{eq:CT toeplitz} and Toeplitz matrices~\eqref{eq:DT toeplitz} are closely related to the convolutions in Section~\ref{sec:LTIintro}, the convolution \eqref{eq:def conv group} can be viewed as a linear operator $\calX:L_2(\setG)\rightarrow L_{\revise{\infty}}(\setG)$ that computes the convolution between the input function $y(g)$ and $x(g)$:
\[
(\calX y)(g) = \int_{\setG} x(h^{-1}\circ g) y(h) \dif h.
\]
We refer to $\calX$ as a {\em Toeplitz operator} since this linear operator involves a kernel function $x(h^{-1}\circ g)$ that depends only on the difference $h^{-1}\circ g$.
We call $\widehat x(\xi)$, the Fourier transform of $x(g)$, the {\em symbol} corresponding to the Toeplitz operator $\calX$.

Finally, let $\setA\in \setG$ be a subset of $\setG$. As explained in Section~\ref{sec:effects of time-limiting}, we are also interested in the {\em time-limited Toeplitz operator}\footnote{$\calX_{\setA}$ is also referred to as a Toeplitz operator in \cite{hirschman1982szegHo,linnik1976canonical,krieger1965toeplitz,grenander1958toeplitz}. \revise{Again, we note that here ``time" refers to the domain in $\setG$.}} $\calX_{\setA}:L_2(\setG)\rightarrow L_2(\setG)$, where
\begin{align}
(\calX_{\setA} y)(g) = \begin{cases}\int_{\setA} x(h^{-1}\circ g) y(h) \dif h, & g\in \setA,\\
 0, & \text{otherwise}. \end{cases}
\label{eq:def Toeplitz operator}\end{align}
Recall that the spectrum of a linear operator $\calX$ is the set of complex numbers $\lambda$ such that $\calX - \lambda \calI$ (where $\calI$ denotes the identity operator) is not invertible. Here the time-limited Toeplitz operator $\calX_{\setA}$ is compact and thus has a discrete spectrum containing what refer to as its eigenvalues. There is no simple formula for exactly expressing the eigenvalues of $\calX_{\setA}$. Instead, we are interested in questions such as: How does the spectrum of the time-limited Toeplitz operator $\calX_{\setA}$ relate to the spectrum of the original (non-time-limited) Toeplitz operator $\calX$? In what sense do the eigenvalues converge as the domain $\setA$ of time-limiting approaches the entire group $\setG$? We discuss answers to these questions in \Cref{sec:General Toeplitz Operators} and \Cref{sec:TF operator on group}.

\subsection{The effective dimensionality of a signal family}
\label{sec:number of degrees}

One of the useful applications of characterizing the spectrum of time-limited Toeplitz operators is in computing the effective dimensionality (or the number of degrees of freedom) of certain related signal families. In this section, we formalize this notion of effective dimensionality for a set of functions defined on a group $\setG$.

\subsubsection{Definitions}

We begin by defining general sets of time-limited functions that we are interested in; in later sections we discuss these functions in the context of communications and signal processing.
Specifically, suppose $\setA$ is a subset of $\setG$ and let $\setW(\setA,\widehat\phi(\xi))\subset L_2(\setA)$ denote the set of functions controlled by a symbol $\widehat\phi(\xi)$:
\begin{align}
\setW(\setA,\widehat\phi(\xi)) := \left\{x\in L_2(\setA): x(g) = \int_{\widehat \setG} \alpha(\xi)\widehat\phi(\xi)\chi_\xi(g) \dif \xi, \int \left|\alpha(\xi)\right|^2 \dif \xi\leq 1, g\in \setA \right\},
\label{eq:set W}\end{align}
which is a subset of $L_2(\setA)$. We note that in \eqref{eq:set W}, the symbol $\widehat\phi(\xi)$ is fixed and we discuss its role soon.

Also let $\setM_n\subset L_2(\setG)$ denote an $n$-dimensional subspace of $L_2(\setG)$. The distance between a point $x\in L_2(\setG)$ and the subspace $\setM_n$ is defined as
\begin{align}
d(x, \setM_n) := \inf_{y\in \setM_n}\int \left( x(g)  - y(g) \right)^2 \dif g = \int \left( x(g)  - (\P_{\setM_n}x)(g) \right)^2 \dif g = \sup_{z\in L_2(\setG), z\perp \setM_n} \frac{\left|\left\langle x, z \right\rangle_{L_2(\setG)}\right| }{\|z\|_{L_2(\setG)}},
\label{eq:distance}\end{align}
where $\P_{\setM_n}:L_2(\setG)\rightarrow L_2(\setG)$ represents the orthogonal projection onto the subspace $\setM_n$.
We define the \revise{width} $d(\setW(\setA,\widehat\phi(\xi)),\setM_n)$ of the set $\setW(\setA,\widehat\phi(\xi))$ \revise{with respect to} the subspace $\setM_n$ as follows:
\[
d(\setW(\setA,\widehat\phi(\xi)),\setM_n):=\sup_{x\in\setW(\setA,\widehat\phi(\xi))} d(x,\setM_n) = \sup_{x\in\setW(\setA,\widehat\phi(\xi))} \inf_{y\in \setM_n}\int \left( x(g)  - y(g) \right)^2 \dif g,
\]
which also represents the largest distance from the elements in $\setW(\setA,\widehat\phi(\xi))$ to the $n$-dimensional subspace $\setM_n$. The Kolmogorov $n$-width~\cite{Kolmogorov36}, $d_n(\setW(\setA,\widehat\phi(\xi)))$ of $\setW(\setA,\widehat\phi(\xi))$ in $L_2(\setG)$ is defined as the smallest \revise{width} $d(\setW(\setA,\widehat\phi(\xi)),\setM_n)$ over all $n$-dimensional subspaces of $L_2(\setG)$; that is
\begin{align}
d_n(\setW(\setA,\widehat\phi(\xi))) : = \inf_{\setM_n}d(\setW(\setA,\widehat\phi(\xi)),\setM_n).
\label{eq:def n-width}\end{align}

In summary, the $n$-width $d_n(\setW(\setA,\widehat\phi(\xi)))$ characterizes how well the set $\setW(\setA,\widehat\phi(\xi))$ can be approximated by an $n$-dimensional subspace of $L_2(\setG)$. By its definition, $d_n(\setW(\setA,\widehat\phi(\xi)))$ is non-increasing in terms of the dimensionality $n$. For any fixed $\epsilon>0$, we define the {\em effective dimensionality}, or {\em number of degrees of freedom}, of the set $\setW(\setA,\widehat\phi(\xi))$ at level $\epsilon$ as~\cite{franceschetti2015landau}
\begin{align}
\calN(\setW(\setA,\widehat\phi(\xi)),\epsilon) = \min\left\{n: d_n(\setW(\setA,\widehat\phi(\xi)))< \epsilon  \right\}.
\label{eq:DoF}\end{align}
In words, the above definition ensures that there exists a subspace $\setM_n$ of dimension $n=\calN(\setW(\setA,\widehat\phi(\xi)),\epsilon)$ such that for every function $x \in \setW(\setA,\widehat\phi(\xi))$, one can find at least one function $y \in \setM_n$ so that the distance between $x$ and $y$ is at most $\epsilon$.

We note that the reason we impose an energy constraint on the elements $x$ of $\setW(\setA,\widehat\phi(\xi))$ in~\eqref{eq:set W} is that we use the absolute distance to  quantify the proximity of $x$ to the subspace $\setM_n$ in \eqref{eq:distance}.

\subsubsection{Connection to operators}

In order to compute $\calN(\setW(\setA,\widehat\phi(\xi)),\epsilon)$, we may define an operator $\calA : L_2(\widehat \setG) \rightarrow L_2(\setA)$ as
\[
(\calA \alpha)(g) = \int_{\widehat \setG}\alpha(\xi)\widehat\phi(\xi)\chi_\xi(g)\dif \xi, ~ g \in \setA.
\]
The adjoint $\calA^*:L_2(\setA)\rightarrow L_2(\widehat\setG)$ is given by $
(\calA^*x)(\xi) = \int_{\setA} x(g)\widehat\phi^*(\xi)\chi_\xi^*(g)\dif g.$
The composition of $\calA$ and $\calA^*$ gives a self-adjoint operator $\calA\calA^*:L_2(\setA)\rightarrow L_2(\setA)$ as follows:
\begin{equation}\begin{split}
&(\calA\calA^*x)(g) = \int_{\widehat\setG} \widehat\phi(\xi)\chi_\xi(g) \int_{\setA} x(h)\widehat\phi^*(\xi)\chi_\xi^*(h)\dif h \dif\xi\\
& = \int_{\setA}  x(h)\int_{\widehat\setG} \left|\widehat\phi(\xi)\right|^2 \chi_\xi(h^{-1}\circ g)\dif\xi \dif h = \int_{\setA}  x(h) (\phi \star \phi^*)(h^{-1}\circ g) \dif h,
\end{split}
\label{eq:AAadjoint}\end{equation}
where $\phi(g) = \int_{\widehat \setG} \widehat\phi(\xi)\chi_\xi(g)\dif \xi $ is the inverse Fourier transform of $\widehat \phi$. In words, compared with \eqref{eq:def Toeplitz operator}, the self-adjoint operator $\calA\calA^*$ can be viewed as a time-limited Toeplitz operator with the kernel $\phi \star \phi^*$.

The following result will help in computing $d_n(\setW(\setA,\widehat\phi(\xi)))$ and the effective dimensionality of $\setW(\setA,\widehat\phi(\xi))$ as well as choosing the optimal basis for representing the elements of $\setW(\setA,\widehat\phi(\xi))$.
\begin{prop}\cite{pinkus2012n}
Let the eigenvalues of $\calA\calA^*$ be denoted and arranged as $\lambda_1\geq \lambda_2\geq\cdots$. Then the $n$-width of $\setW(\setA,\widehat\phi(\xi))$ can be computed as
\[
d_n(\setW(\setA,\widehat\phi(\xi))) = \sqrt{\lambda_n},
\]
and the optimal $n$-dimensional subspace to represent $\setW(\setA,\widehat\phi(\xi))$ is the subspace spanned by the first $n$ eigenvectors of $\calA\calA^*$.
\label{prop: n-width}\end{prop}
The proof of Proposition~\ref{prop: n-width} is given in Appendix~\ref{sec:prf prop n-width}.

\section{General Toeplitz Operators on Locally Abelian Groups}
\label{sec:General Toeplitz Operators}

Let $x(g)\in L_2(\setG)$, whose Fourier transform is given by $\widehat x(\xi)$. Now we are well equipped to consider the eigenvalue distribution of a general time-limited Toeplitz operator $\calX_{\setA}$ (which is formally defined in~\eqref{eq:def Toeplitz operator}) on a locally abelian group; in particular we are interested in the relationship between the spectrum of of the time-limited Toeplitz operator $\calX_{\setA}$ and  $\widehat x(\xi)$. The operator $\calX_{\setA}$ is completely continuous and its eigenvalues are denoted by $\lambda_{\ell}(\calX_{\setA})$. Before presenting the main results, we introduce new notation for subsets of $\setG$ (or $\widehat \setG$) which are asymptotically increasing to cover the whole group. This is similar to how we discussed the cases where $T \rightarrow \infty$ and $N \rightarrow \infty$ in Section~\ref{sec:effects of time-limiting}. To that end, let $\setA_\tau, \tau\in(0,\infty)$ be a system of \revise{Borel} subsets of $\setG$ \revise{with boundaries of measure zero} such that $0<\mu(\setA_\tau)<\infty$. The subscript $\tau$ is sometimes dropped when it is clear from the context.
We can view $\setA_\tau$ as a set of subsets that depend on the parameter $\tau$. One can also define a system of  subsets with multiple parameters.

\subsection{Generalized Szeg{\H o}'s theorem}
Abundant effort~\cite{hirschman1982szegHo,linnik1976canonical,krieger1965toeplitz,grenander1958toeplitz}
has been devoted to extending the conventional Szeg{\H o}'s theorem for a general time-limited Toeplitz operator $\calX_{\setA}$. Let $\mathcal{N}(\calX_{\setA};(a,b))=:\#\left\{\ell:  a<\lambda_{\ell}(\calX_{\setA})<b\right\}$ denote the number of eigenvalues of $\calX_{\setA}$ that are between $a$ and $b$. We summarize the following generalized Szeg{\H o}'s theorem concerning the collective behavior of the eigenvalues of $\calX_{\setA}$ and relating them to $\widehat x(\xi)$ (the spectrum of the corresponding non-time-limited operator $\calX$).

\begin{thm}[Generalized Szeg{\H o}'s theorem~\cite{hirschman1982szegHo,linnik1976canonical,krieger1965toeplitz,grenander1958toeplitz}]
Let $x(g)\in L_2(\setG)$ and $\calX_{\setA_\tau}$ be the time-limited Toeplitz operator defined in~\eqref{eq:def Toeplitz operator}. Suppose the following holds almost everywhere:
\[
\lim_{\tau\rightarrow \infty}\setA_\tau = \setG.
\]
 Then for all intervals $(a,b)$ such that $\nu\left(\left\{\xi:\widehat x(\xi) = a\right\}\right) = \nu\left(\left\{\xi:\widehat x(\xi) = b\right\}\right) = 0$, we have
\begin{align}
\lim_{\tau\rightarrow \infty} \frac{\calN(\calX_{\setA_\tau};(a,b))}{\mu(\setA_\tau)} = \nu\left(\left\{\xi: a<\widehat x(\xi)<b\right\}\right).
\label{eq:general szego}\end{align}
\label{thm:general szego}\end{thm}
In a nutshell, Theorem~\ref{thm:general szego} implies that the eigenvalues of the time-limited Toeplitz operator $\calX_{\setA_\tau}$ are closely related to $\widehat x(\xi)$, the spectrum of the corresponding non-time-limited Toeplitz operator $\calX$. Some work instead presents~\eqref{eq:general szego} as
\begin{align}
\lim_{\tau\rightarrow \infty} \frac{\calN(\calX_{\setA_\tau};[a,b])}{\mu(\setA_\tau)} = \nu\left(\left\{\xi: a\leq \widehat x(\xi)\leq b\right\}\right).
\label{eq:general szego 2}\end{align}
One can understand the equivalence between~\eqref{eq:general szego} and \eqref{eq:general szego 2} as the boundary of the interval makes no difference since
$\nu\left(\left\{\xi:\widehat x(\xi) = a\right\}\right) = 0$ and $\nu\left(\left\{\xi:\widehat x(\xi) = b\right\}\right) = 0$.

In words, Theorem~\ref{thm:general szego} implies that the eigenvalue distribution of the operator $\calX_{\setA_\tau}$ asymptotically converges to the distribution of the Fourier transform of $x(g)$. We now compare Theorem~\ref{thm:general szego} with the conventional Szeg{\H o}'s theorems in Section~\ref{sec:conventional Szego}
that have widely appeared in information theory and signal processing. We note that \eqref{eq:general szego} has exactly the same form as \eqref{eq:szego continuous} for the time-limited operator $\calH_T$ in \eqref{eq:CTTL}.

For the Toeplitz matrix $\mH_N$ defined in \eqref{eq:FDT toeplitz}, at first glance, \eqref{eq:Szego thm} is slightly different than what is expressed in~\eqref{eq:general szego} which implies
\begin{align}
\lim_{N\rightarrow\infty}\frac{\calN(\mH_N;(a,b))}{N} = \left|\left\{f: f\in [0,1), a<\widehat h(f)<b\right\}\right|
\label{eq:Szego thm CDF}\end{align}
for all intervals $(a,b)$ such that $\left|\left\{f:\widehat h(f) = a\right\}\right| = 0$ and $\left|\left\{f:\widehat h(f) = b\right\} \right|= 0$. In fact, \eqref{eq:Szego thm} and \eqref{eq:Szego thm CDF} are equivalent if we view $\widehat h:[0,1)\rightarrow \R$ as a random variable and define $\vlambda_{\mH_N}:\left\{0,1,\ldots,N-1\right\}\rightarrow \R$ to be a discrete random variable by $\vlambda_{\mH_N}[\ell] = \lambda_\ell(\mH_N)$. In probabilistic language, set
\[
F_{\widehat h}(a):= \left|\left\{f: f\in [0,1), \widehat h(f)\leq a\right\}\right|
 \]
to be the cumulative distribution function (CDF) associated to $\widehat h$. Also denote the CDF associated to $\vlambda_{\mH_N}$ by
\[
F_{\vlambda_{\mH_N}}(a):=\frac{\calN\left(\mH_N;(-\infty,a]\right)}{N} = \frac{\#\left\{\ell:\lambda_{\ell}(\mH_N)\leq a \right\}}{N}.
\]
The following result, known as the Portmanteau lemma, gives two equivalent descriptions of weak convergence in terms of the CDF and the means of the random variables.
\begin{lem}\cite[Portmanteau lemma]{van2000asymptotic}
The following are equivalent:
\begin{enumerate}
\item$\lim_{N\rightarrow \infty}\frac{1}{N}\sum_{\ell=0}^{N-1}\vartheta(\lambda_\ell(\mH_N)) = \int_{0}^{1}\vartheta(\widehat{h}(f))df$ for all bounded, continuous functions
     $\vartheta$;
\item $\lim_{N\rightarrow \infty}F_{\vlambda_{\mH_N}}(a) = F_{\widehat h}(a)$
for every point $a$ at which $F_{\widehat h}$ is continuous.
\end{enumerate}
\label{lem:convergence cdf}\end{lem}
Note that if  $F_{\widehat h}$ is continuous at $a$, then $\left|\left\{f:\widehat h(f) = a\right\}\right| = 0$. Thus the equivalence between \eqref{eq:Szego thm} and \eqref{eq:Szego thm CDF} follows from the Portmanteau lemma. In words, \eqref{eq:Szego thm CDF}  implies that certain collective behaviors of the eigenvalues of each Toeplitz matrix are reflected by the symbol $\widehat h(f)$.

We note that \eqref{eq:Szego thm} is one of the descriptions of weak convergence of a sequence of random variables in the Portmanteau lemma~\cite{van2000asymptotic} (also see Lemma~\ref{lem:convergence cdf}). Thus, throughout the paper, we also refer to the collective behavior (like that characterized by~\eqref{eq:Szego thm}) of the eigenvalues as the {\em distribution} of the eigenvalues.


In the following subsections, we discuss applications of the generalized Szeg\H{o}'s theorem.

\subsection{Application: Subspace approximations}

\subsubsection{Convolutions with a pulse}
We first consider the set of functions obtained by time-limiting the convolution between $\alpha(g)$ and a fixed function $\phi(g)$:
\begin{align}
\setW(\setA,\phi(g)): = \left\{x\in L_2(\setA): x(g) = \int_{\setG} \alpha(h)\phi(h^{-1}\circ g) \dif h, \int_\setG \left|\alpha(g)\right|^2 \dif g\leq 1, g\in \setA \right\}.
\label{eq:setW conv}\end{align}
We note that $\setW(\setA,\phi(g))$ is equivalent to $\setW(\setA,\widehat\phi(\xi))$ defined in \eqref{eq:set W} by rewriting $x(g)$ in \eqref{eq:setW conv}:
\begin{align*}
x(g) = \int_{\setG} \alpha(h)\phi(h^{-1}\circ g) \dif h = \int_{\widehat\setG} \widehat \alpha(\xi)\widehat \phi(\xi)\chi_\xi(g)\dif \xi,
\end{align*}
which is exactly the same form of $x(g)$ in \eqref{eq:set W}. This model \eqref{eq:setW conv} arises in  radar signal processing, channel sensing, and super-resolution of pulses through an unknown channel. Proposition~\ref{prop: n-width} implies that the $n$-width of $\setW(\setA,\phi(g))$ is given by
\[
d_n(\setW(\setA,\phi(g))) = \sqrt{\lambda_n(\calA\calA^*)},
\]
where $\calA\calA^*$ defined in \eqref{eq:AAadjoint} is  time-limited Toeplitz operator with the kernel $\phi\star\phi^*$. Now Theorem~\ref{thm:general szego} along with \eqref{eq:DoF} reveals  the effective dimensionality, or number of degrees of freedom,
of the set $\setW(\setA,\phi(g))$ at level $\epsilon$ as follows.
\begin{cor}
Suppose $
\lim_{\tau\rightarrow \infty}\setA_\tau = \setG$
holds almost everywhere. Then for any $\epsilon>0$  such that $\nu\left(\left\{\xi:\widehat \phi(\xi) = \epsilon\right\}\right) = 0$, we have
\begin{align*}
\lim_{\tau\rightarrow\infty}\frac{\calN(\setW(\setA_\tau,\phi(g)),\epsilon)}{\mu(\setA_\tau)} = \nu\left(\left\{\xi: \left|\widehat\phi(\xi)\right|> \epsilon  \right\}\right).
\end{align*}
\label{eq:Cor convolutions with a pulse}\end{cor}
\begin{proof}
By \eqref{eq:DoF}, we have $\calN(\setW(\setA_\tau,\phi(g)),\epsilon) = \min\left\{n:\sqrt{\lambda_n}<\epsilon\right\} = \# \left\{n:\sqrt{\lambda_n}\geq \epsilon\right\}$. Note that the Fourier transform of $\phi\star\phi^*$ is $\left|\widehat \phi(\xi)\right|^2$. Then Corollary~\ref{eq:Cor convolutions with a pulse} then follows directly by applying Theorem~\ref{thm:general szego} to $\calA\calA^*$ (which is a time-limited Toeplitz operator with the kernel $\phi\star\phi^*$).
\end{proof}

\subsubsection{Shifts of a signal}
We now consider a slightly different model where the function of interest is a linear combination of continuous shifts of a given signal $\phi(g)\in L_2(\setG)$:
\begin{align}
\setS(\setA,\phi(g)): = \left\{x\in L_2(\setA): x(g) = \int_{\setA} \alpha(h)\phi(h^{-1}\circ g) \dif h, \int_\setA \left|\alpha(g)\right|^2 \dif g\leq 1\right\}.
\label{eq:setS}\end{align}
Define $\calT_\setA: L_2(\setG)\rightarrow L_2(\setG)$ as a time-limiting operator that makes a function zero outside $\setA$. We can also rewrite $x(g) = \int_{\setG} (\calT_{\setA}\alpha)(h) \phi(h^{-1}\circ g) \dif h = (\calT_{\setA}\alpha) \star \phi$ for $g\in \setA$, the convolution between the time-limited function $(\calT_{\setA}\alpha)(g)$ and $\phi(g)$. By zero-padding the signal $x$ outside $\setA$, we may also rewrite it simply as $x = \calT_{\setA} ((\calT_{\setA}\alpha) \star \phi)$. Analogously, we could express the function $x(g)$ in \eqref{eq:setW conv} as $x = \calT_{\setA}(\alpha \star \phi)$. Now it is clear the model for $\setW(\setA,\phi(g))$ and the one for $\setS(\setA,\phi(g))$ differs in the location of the time-limiting operator $\calT_\setA$.

To investigate the effective dimension or the number of degrees of freedom for $\setS(\setA,\phi(g))$, we define the operator $\calS: L_2(\setA)\rightarrow L_2(\setG)$ as:
\[
(\calS \alpha)(g) = \int_{\setA}\alpha(h)\phi(h^{-1}\circ g)\dif h, \ g\in\setG.
\]
Its adjoint $\calS^*:L_2(\setG)\rightarrow L_2(\setA)$  is given by
\[
(\calS^* x)(g) = \int_{\setG}\phi^*(g^{-1}\circ h)x(h)\dif h, \ g\in\setA.
\]
We then have the self-adjoint operator $\calS^*\calS:L_2(\setA)\rightarrow L_2(\setA)$
\begin{align*}
&(\calS^*\calS \alpha)(g) = \int_{\setG} \phi^*(g^{-1}\circ h)\int_{\setA}\alpha(\eta)\phi(\eta^{-1}\circ h)\dif \eta \dif h\\
& = \int_{\setA} \int_{\setG} \phi^*(g^{-1}\circ h) \phi(\eta^{-1}\circ h)\dif h \alpha(\eta)\dif \eta = \int_{\setA} r(\eta^{-1}\circ g) \alpha(\eta)\dif \eta,
\end{align*}
where $r(g): = \int_{\setG}\phi(h)\phi(g^{-1}\circ h)\dif h$ is the autocorrelation function of the function $\phi$. Thus, $\calS^*\calS$ is a Toeplitz operator of the form \eqref{eq:def Toeplitz operator}. Similar to Proposition~\ref{prop: n-width}, we can study the effective dimension of the set of shifted signals in \eqref{eq:setS} by looking at the eigenvalue distribution of the self-adjoint operator $\calS\calS^*$. Note that in this case $\calS\calS^*$ is not a Toeplitz operator, but it has the same nonzero eigenvalues as $\calS^*\calS$. Thus, we can exploit the eigenvalue distribution of $\calS^*\calS$ to infer the number of degrees of freedom for the set $\setS(\setA,\phi(g))$. This is formally established in the following result.
\begin{prop}
Let the eigenvalues of $\calS^*\calS$ be denoted and arranged as $\lambda_1\geq \lambda_2\geq\cdots$. Then the $n$-width of $\setS(\setA,\phi(g))$ can be computed as $
d_n(\setS(\setA,\phi(g))) = \sqrt{\lambda_n},$
and the optimal $n$-dimensional subspace to represent $\setS(\setA,\phi(g))$ is the subspace spanned by the first $n$ eigenvectors of $\calS\calS^*$.
\label{prop: n-width shift manifold}\end{prop}
Finally, \Cref{thm:general szego} implies that the eigenvalue distribution of the Toeplitz operator $\calS^*\calS$ is asymptotically equivalent to $\widehat r(\xi) = \int_{\setG}r(g)\chi_\xi(g)\dif g$, the power spectrum of $\phi$ if we view $r$ as the autocorrelation of $\phi$.

\subsection{Application: Eigenvalue estimation}

In many applications such as spectrum sensing algorithm for cognitive radio~\cite{zeng2009eigenvalueSpectrumSensing}, it is desirable to understand the {\em individual} asymptotic behavior of the eigenvalues of time-limited Toeplitz operators rather than the collective behavior of the eigenvalues provided by Szeg\H{o}'s theorem (Theorem~\ref{thm:general szego}). As a special case, efficiently estimating the spectral norm (i.e., the largest singular value) of Toeplitz matrices is crucial in certain applications. For example, the Lipschitz constant of a CNN has wide implications in understanding the key properties of the neural network such as its generalization and robustness. Unfortunately, computing the exact Lipschitz constant of a neural network is known to be NP-hard~\cite{scaman2018lipschitz}. Recent work~\cite{yi2020asymptotic,araujo2021lipschitz} proposed methods for computing upper bounds of the Lipschitz constant for each layer (and hence for the entire network) by efficiently estimating the spectral norm of the corresponding (block) Toeplitz matrices. We will review related recent progress on characterizing the individual behavior of the eigenvalues for Toeplitz matrices. To our knowledge, the individual behavior of the eigenvalues has only recently been investigated for Toeplitz matrices.

Bogoya et al.~\cite{Bogoya2015maximum} studied the individual asymptotic behavior of the eigenvalues of Toeplitz matrices by interpreting Szeg\H{o}'s theorem in \eqref{eq:Szego thm} in probabilistic language and related the eigenvalues to the values obtained by sampling $\widehat h(f)$ uniformly in frequency on $[0,1)$:
\begin{align}
\lim_{N\rightarrow \infty} \max_{0\leq\ell\leq N-1}\left|\lambda_\ell(\mH_N) - \widehat h(\frac{\rho(\ell)}{N})\right| = 0
\label{eq:individual behav}\end{align}
if the range of $\widehat h(f)$ is connected. Here $\widehat h(\frac{\rho(\ell)}{N})$  is the permuted form of $\widehat h(\frac{\ell}{N})$ such that $\widehat h(\frac{\rho(0)}{N}) \geq \widehat h(\frac{\rho(1)}{N}) \geq\cdots \geq \widehat h(\frac{\rho(\ell)}{N})$. Thus, if the symbol $\widehat h(f)$ is known, we can sample it uniformly to get  reasonable estimates for the eigenvalues of the Toeplitz matrix.

Despite the power of Szeg\H{o}'s theorem, in many scenarios (such as certain coding and filtering applications~\cite{gray1972asymptotic,pearl1973coding}), one may only have access to $\mH_N$ and not $\widehat{h}$. In such cases, it is still desirable to have practical and efficiently computable estimates of the individual eigenvalues of $\mH_N$. We recently showed~\cite{zhu2017asymptotic} that we can construct a certain sequence of $N\times N$ circulant matrices such that the eigenvalues of the circulant matrices asymptotically converge to those of the Toeplitz matrices. Transforming the Toeplitz matrix into a circulant matrix can be performed extremely efficiently using closed form expressions; the eigenvalues of the circulant matrix can then be computed very efficiently (in $O(N\log N)$ using the fast Fourier transform (FFT)).

When the sequence $h[n]$ is not symmetric about the origin, the Avram-Parter theorem~\cite{avram1988bilinear,parter1986distribution}, a generalization of Szeg\H{o}'s theorem, relates the collective asymptotic behavior of the singular values of a general (non-Hermitian) Toeplitz matrix to the absolute value of its symbol, i.e., $|\widehat h(f)|$. Bogoya et al.~\cite{Bogoya2015maximum} also showed that the singular values of $\mH_N$ asymptotically converge to the uniform samples of $|\widehat h(f)|$ provided  the range of the symbol $|\widehat h(f)|$ is connected.

\subsection{Questions}
Inspired by the applications listed above, we raise two questions concerning the generalized Szeg{\H o}'s theorem (Theorem~\ref{thm:general szego}). The first question concerns the individual behavior of the eigenvalues.
\begin{question}
Is it possible to extend the result \eqref{eq:individual behav} concerning the individual behavior of the eigenvalues for Toeplitz matrices to general Toeplitz operators?
\end{question}

We note that both the conventional Szeg{\H o}'s theorem listed in Section~\ref{sec:conventional Szego} and the generalized Szeg{\H o}'s theorem (Theorem~\ref{thm:general szego}) characterize asymptotic behavior of the eigenvalues.
\begin{question}
Is it possible to establish a non-asymptotic result concerning the eigenvalue behavior (either collective or individual) for the general Toeplitz operators $\calX_{\setA_\tau}$?
\end{question}

\section{Time-Frequency Limiting Operators on Locally Compact Abelian Groups}
\label{sec:TF operator on group}

In this section, we consider a special case of time-limited Toeplitz operators: time-frequency limiting operators on locally compact abelian groups, to be formally defined soon. As we have briefly explained in Section~\ref{sec:effects of time-limiting}, time-frequency limiting operators in the context of the classical groups where $\setG$ are the real-line, $\Z$, and $\Z_N$ play important roles in signal processing and communication. By considering time-frequency limiting operators on locally compact abelian groups, we aim to ($i$)~provide a unified treatment of the previous results on the eigenvalues of the operators resulting in PSWFs, DPSSs, and periodic DPSSs (PDPSSs)~\cite{jain1981extrapolation,grunbaum1981eigenvectors}; and ($ii$)~extend these results to other signal domains such as rotations in a plane and three dimensions~\cite[Chapter 5]{chirikjian2016harmonic}. In particular, we will investigate the eigenvalues of time-frequency limiting operators on locally compact abelian groups and show that they exhibit similar behavior to both the conventional CT and DT settings: when sorted by magnitude, there is a cluster of eigenvalues close to (but not exceeding) 1, followed by a \revise{relatively sharp} transition, after which the remaining eigenvalues are close to 0. This behavior also resembles the rectangular shape of the frequency response of the original band-limiting operator. Although this collective behavior can be characterized using Szeg\H{o}'s theorem, finer results (particularly non-asymptotic results) on the eigenvalues have been established in special cases, which we will review in detail. We will also discuss the applications of this unifying treatment  in relation to channel capacity and to representation and approximation of signals.

To introduce the time-frequency limiting operators, consider two subsets $\setA\in \setG$ and $\setB\in \widehat \setG$. Recall that $\calT_\setA: L_2(\setG)\rightarrow L_2(\setG)$ is a time-limiting operator that makes a function zero outside $\setA$. 
Also define $\calB_\setB = \calF^{-1}\calT_\setB\calF: L_2(\setG)\rightarrow L_2(\setG)$ as a band-limiting operator that takes the Fourier transform of an input function on $L_2(\setG)$, sets it to zero outside $\setB$, and then computes the inverse Fourier transform. The operator $\calB_\setB$ acts on $L_2(\setG)$ as a convolutional integral operator:
\begin{align*}
(\calB_\setB x)(g) &= \int_\setB \widehat x(\xi) \chi_\xi(g)   \dif \xi = \int_\setB \left(\int_\setG x(h)\chi_\xi^*(h)\dif h\right) \chi_\xi(g)   \dif \xi = \int_\setG K_\setB(h^{-1}\circ g)   x(h) \dif h,
\end{align*}
 where
\begin{align}
K_\setB(h^{-1}\circ g) = \int_\setB \chi_\xi^*(h) \chi_\xi(g) \dif \xi = \int_\setB \chi_\xi(h^{-1}\circ g) \dif \xi.
\label{eq:def KB}\end{align}
It is of interest to study the eigenvalues of the following operators which we refer to as {\em time-frequency limiting operators}
\begin{align}
\calO_{\setA,\setB} = \calT_\setA\calB_\setB\calT_\setA, \quad ~\text{and}  \quad \calB_\setB\calT_\setA\calB_\setB.
\label{eq:generalized time-band limiting operator}\end{align}
Utilizing the expression for $\calB_\setB$, the operator $\calT_\setA\calB_\setB\calT_\setA$ acts on any $x\in L_2(\setG)$ as follows
\[
\left(\calT_\setA\calB_\setB\calT_\setA x\right)(g) =
\begin{cases}\int_\setA K_\setB(h^{-1}\circ g)   x(h) \dif h, & g\in \setA\\
 0, & \text{otherwise}. \end{cases}
\]
The operator $\calO_{\setA,\setB}$ is symmetric and completely continuous and we denote its eigenvalues by $\lambda_{\ell}(\calO_{\setA,\setB})$.
Due to the time- and band-limiting characteristics of the operator $\calO_{\setA,\setB}$, the eigenvalues of $\calO_{\setA,\setB}$ are between $0$ and $1$. To see this, let $x(g)\in L_2(\setA)$:
\begin{align*}
&\left\langle (\calO_{\setA,\setB}x)(g), x(g) \right\rangle = \left\langle \int_\setA\int_\setB \chi_\xi(h^{-1}\circ g) \dif \xi   x(h) \dif h, x(g) \right\rangle  \\ &= \int_\setB \left( \int_\setA\int_\setA \chi_\xi(h^{-1}\circ g)    x(h)  x^*(g) \dif h \dif g\right) \dif \xi = \int_\setB   \left|\widehat x(\xi)\right|^2  \dif \xi \geq 0.
\end{align*}
On the other hand, we have $
\int_\setB   \left|\widehat x(\xi)\right|^2  \dif \xi \leq \int_{\widehat\setG}   \left|\widehat x(\xi)\right|^2  \dif \xi = \int_{\setG}   \left|x(g)\right|^2  \dif g.$


\subsection{Eigenvalue distribution of time-frequency limiting operators}
To investigate the eigenvalues of the operator $\calO_{\setA,\setB} = \calT_\setA\calB_\setB\calT_\setA$, we first note that without the time-limiting operator $\calT_\setA$, the eigenvalues of $\calB_\setB$ are simply given by the Fourier transform of $K_\setB(g)$, and thus they are either 1 or 0. Our main question is how the spectrum of the time-frequency limiting operator relates to the spectrum of the band-limited operator. Based on the binary spectrum of $\calB_\setB$ and the intuition from Szeg{\H o}'s theorem in \Cref{thm:general szego}, we expect that the eigenvalues of $\calO_{\setA,\setB}$ to have a particular behavior: when sorted by magnitude, there should be a cluster of eigenvalues close to (but not exceeding) $1$, followed by an abrupt transition, after which the remaining eigenvalues should be close to $0$.  Moreover, the number of effective (i.e., relatively large) eigenvalues should be essentially equal to the time-frequency area $|\setA||\setB|$. These results are confirmed below and reveal the dimensionality (or the number of degrees of freedom) of classes of band-limited signals observed over a finite time, which is fundamental to characterizing the performance limits of communication systems.

We note that similar to how we discussed the cases where $T \rightarrow \infty$ and $N \rightarrow \infty$ in Section~\ref{sec:effects of time-limiting}, we will use $\setA_\tau, \tau\in(0,\infty)$ to define the subsets of $\setG$ that depend on $\tau$. The subscript $\tau$ is often dropped when it is clear from the context. We now present one of our main results concerning the asymptotic behavior for the eigenvalues of the time-frequency limiting operators $\calO_{\setA_\tau,\setB}$ when $\setA_\tau$ approaches $\setG$.
\begin{thm} Suppose $\setB$ is a fixed subset of $\widehat\setG$ and let $\epsilon \in (0,\frac{1}{2})$.
 Let
 \[
 \calN(\calO_{\setA_\tau,\setB};(a,b)):=\#\left\{\ell:a<\lambda_{\ell}(\calO_{\setA_\tau,\setB})<b\right\}
 \]
 denote the number of eigenvalues of $\calT_{\setA_\tau}\calB_\setB\calT_{\setA_\tau}$ that are between $a$ and $b$. Then if
\begin{align}
\lim_{\tau\rightarrow \infty}\setA_\tau = \setG
\label{eq:cond A_tau}\end{align}
holds almost everywhere, we have
\begin{align}
\sum_{\ell} \lambda_{\ell}(\calO_{\setA_\tau,\setB}) = |\setA_\tau| |\setB|,\quad
\sum_{\ell}\lambda^2_{\ell}(\calO_{\setA_\tau,\setB}) = |\setA_\tau| |\setB| - o(|\setA_\tau| |\setB|),\label{eq:sum of squares}
\end{align}
and
\begin{align}
\lim_{\tau\rightarrow \infty}\frac{\calN(\calO_{\setA_\tau,\setB};[1-\epsilon,1])}{|\setA_\tau|} = |\setB|, \quad \calN(\calO_{\setA_\tau,\setB};(\epsilon,1-\epsilon)) = o\left(\frac{|\setA_\tau| |\setB|}{\epsilon(1-\epsilon)}\right).\label{eq:transition region}
\end{align}
\revise{Here $|\cdot|$ denotes the Haar measure.}
\label{thm:limiting operator}\end{thm}
The proof of Theorem~\ref{thm:limiting operator} is given in Appendix~\ref{sec:prf thm limiting operator}. \revise{The limit in \eqref{eq:cond A_tau} is in the sense of convergence in measure on each compact set of $\setG$. There are no specific shape constraints in $\setA_\tau$ except that their boundaries are measure zero, though we note that in many cases of interest $\setA_\tau$ is a closed set as in bandlimited signals \cite{landau1975szego}, a union of closed sets as in multiband signals \cite{zhu2015approximating}, or a scaling of a fixed set as in \cite{franceschetti2015landau} where $\mQ_\tau\setA = \{\mQ_\tau g: g\in\setA\}$ where $\setG$ is $\R^n$, $\setA$ is a fixed set of $\R^{n}$ with boundary of measure zero, and $\mQ_\tau\in\R^{n\times n}$ depends on $\tau$ such that $\lim_{\tau\rightarrow\infty} \mQ_{\tau}\setA = \R^n$.} Theorem~\ref{thm:limiting operator} formally confirms that the spectra of the time-frequency limiting operators resemble the rectangular shape of the spectrum of the band-limiting operator. As guaranteed by~\eqref{eq:transition region}, the number of effective eigenvalues of the time-frequency limiting operator is asymptotically equal to the time-frequency area $|\setA_\tau||\setB|$. Similar results for time-frequency limiting operators in the context of classical groups where $\setG$ is $\R^{n}$ are given in \cite{landau1975szego,franceschetti2015landau}. We discuss the applications of Theorem~\ref{thm:limiting operator} in  channel capacity and representation and approximation of signals in more detail in the following subsections.

As mentioned before, the time-frequency limiting operators in the context of the classical groups where $\setG$ are the real-line, $\Z$, and $\Z_N$ were first studied by Landau, Pollak, and Slepian who wrote a series of papers regarding the dimensionality of time-limited signals that are approximately band-limited (or vice versa)~\cite{SlepiP_ProlateI,LandaP_ProlateII,LandaP_ProlateIII,Slepi_ProlateIV,Slepian78DPSS} (see also~\cite{Slepi_On,Slepi_Some} for concise overviews of this body of work). After that, a set of results concerning the number of eigenvalues within the transition region $(0,1)$ have been established in~\cite{landau1980eigenvalue,EdelmanMT1998FutureFFT,osipov2013certain,Israel2015EigenvalueDisTFLocalization,Karnik2016FAST,zhu2017DFT}. which will be reviewed in detail in the following remarks.

\begin{remark}
Using the explicit expressions for the character function $\chi_\xi(g)$ and the kernel $K_\setB(g)$ and applying integration by parts for \eqref{eq:sum of squares KB}, one can improve the second term in \eqref{eq:sum of squares} to $O(\log(|\setA_\tau||\setB|))$ for many common one-dimensional cases:

\begin{itemize}[leftmargin=*]
\item  Suppose $\setG = \R$ and $\widehat \setG = \R$. Let $\setA_T = [-\frac{T}{2},\frac{T}{2}]$ (where $\tau = T$ in \eqref{eq:cond A_tau}) and $\setB = [-\frac{1}{2},\frac{1}{2}]$ without loss of generality. Then the kernel $K_\setB(t)$ turns out to be $
    K_\setB(t) = \int_{-\frac{1}{2}}^{\frac{1}{2}} e^{j2\pi Ft}\dif F = \frac{2\sin(\pi t)}{\pi t}.$
Plugging in this form into \eqref{eq:sum of squares KB} gives~\cite{hogan2011bookDurationBandLimiting} $
    \sum_\ell(\lambda_{\ell}(\calO_{\setA_T,\setB}))^2 = T -O(\log(T)).$
In this case, the operator $\calO_{\setA_\T,\setB}$ is equivalent to the time-limited Topelitz operator $\calH_T$ in Section~\ref{sec:tfintro} and the corresponding eigenfunctions are known as PSWFs.

\item As an another example, suppose $\setG = \Z$, $\widehat \setG = [-\frac{1}{2},\frac{1}{2}]$ and let $\setA_N = \{0,1\ldots,N-1\}$ (where $\tau = N$ in \eqref{eq:cond A_tau}), $\setB = [-W,W]$ with $W\in(0,\frac{1}{2})$. In this case, the kernel $K_\setB(n)$ becomes $
    K_\setB[n] = \int_{-W}^{W} e^{j2\pi fn}\dif f = \frac{\sin(2\pi Wn)}{\pi Wn}.$
Then plugging in this form into \eqref{eq:sum of squares KB} gives~\cite[Theorem 3.2]{zhu2015approximating}
 $
 \sum_\ell(\lambda_{\ell}(\calO_{\setA_N,\setB}))^2 = 2NW - O(\log(2NW)).
 $
We note that in this case, the operator $\calO_{\setA_N,\setB}$ is equivalent to the $N\times N$ prolate matrix $\mB_{N,W}$ with entries
 \begin{align}
 \mB_{N,W}[m,n] := \frac{\sin\left(2\pi W(m-n)\right)}{\pi(m-n)}
 \label{eq:prolate matrix B_NW}\end{align}
for all $m,n\in\{0,1,\ldots,N-1\}$. The eigenvalues and eigenvectors of the matrix $\mB_{N,W}$ are referred to as the DPSS eigenvalues and DPSS vectors, respectively.

\item As a final example, we consider $\setG =\Z_N$, $\widehat\setG = \Z_N$ and the Fourier transform is the conventional DFT. Suppose $M,K\leq N$. Let $\tau = M$ in \eqref{eq:cond A_tau} and
\begin{align}
\setA_M = \left\{0,1,\ldots,M-1\right\}, \ \setB = \left\{0,1,\ldots,K-1\right\},
\label{eq:set for DFT limiting}\end{align}
In this case, $\chi_k[n] = e^{j2\pi\frac{nk}{N}}$ and the kernel $K_\setB[n]$ is $
     K_\setB[n] = \sum_{k=0}^{K-1}e^{j2\pi\frac{nk}{N}} = e^{j\pi n\frac{K-1}{N}}\frac{\sin(\pi\frac{nK}{N})}{\sin(\pi\frac{n}{N})}.$
Then plugging in this form into \eqref{eq:sum of squares KB} gives~\cite{EdelmanMT1998FutureFFT,zhu2017DFT}
$
 \sum_\ell(\lambda_{\ell}(\calO_{\setA_M,\setB}))^2 = \frac{MK}{N} - O(\log(\frac{MK}{N} )).$
\end{itemize}
Through the above examples, one may wonder whether we can in general replace the second term in \eqref{eq:sum of squares} by $O(\log(|\setA_\tau||\setB|))$ with a finer analysis of $\sum_\ell(\lambda_{\ell}(\calO_{\setA_\tau,\setB}))^2$. We utilize a two-dimensional example to answer this question in the negative: Suppose $\setG = \Z^2$, $\widehat \setG = [-\frac{1}{2},\frac{1}{2}]\revise{\times [-\frac{1}{2},\frac{1}{2}]}$ and let $\setA = \{0,1\ldots,N-1\}\times \{0,1\ldots,N-1\}, \setB = [-W,W]\times [-W,W]$ with $W\in(0,\frac{1}{2})$. In this case, the kernel $K_{\setB}[n_1,n_2]$ is
\[
    K_\setB[n_1,n_2] = \int_{-W}^{W}\int_{-W}^{W} e^{j2\pi f_1n_1}e^{j2\pi f_2n_2}\dif f_1\dif f_2 = \frac{\sin(2\pi Wn_1)}{\pi Wn_1}\frac{\sin(2\pi Wn_2)}{\pi Wn_2}.
\]
The eigenvectors of the corresponding operator $\calO_{\setA,\setB}$ are known as the two-dimensional DPSSs. For this case, we have $
 \sum_\ell(\lambda_{\ell}(\calO_{\setA_\tau,\setB}))^2 = 4N^2W^2 - O(NW\log(NW)).$
In other words, in this case, we can only improve the second term in \eqref{eq:sum of squares} to $O(NW\log(NW))$ rather than $O(\log(4N^2W^2))$.
\end{remark}

\begin{remark}
We note that the transition region in \eqref{eq:transition region} depends on $\epsilon$ in the form of $\frac{1}{\epsilon (1-\epsilon)}$. A better understanding of the transition region requires further complicated analysis. In the literature, finer results on the transition region are known for several common cases:
\begin{itemize}[leftmargin=*]
\item The results for the eigenvalue distribution of the continuous time-frequency localization operator (where $\setG = \R$, $\widehat \setG = \R$, $\setA_T = [-\frac{T}{2},\frac{T}{2}]$  and $\setB = [-\frac{1}{2},\frac{1}{2}]$) has a rich history. As one example, for any $\epsilon\in(0,1)$, Landau and Widom~\cite{landau1980eigenvalue} provided the following asymptotic result 
$
\calN(\calO_{\setA_T,\setB};[\epsilon,1]) = T + \left(\frac{1}{\pi^2}\log\frac{1-\epsilon}{\epsilon}  \right)\log\frac{\pi T}{2} + o\left(\log\frac{\pi T}{2} \right).
$
This asymptotic result ensures the $O(\log(\frac{1}{\epsilon})\log(T))$ dependence on $\epsilon$ and time-frequency area $T$. Recently, Osipov~\cite{osipov2013certain} proved that
$
\calN(\calO_{\setA_T,\setB};[\epsilon,1]) \leq T + C\log(T)^2\log(1/\epsilon),
$
 where $C$ is a constant.  Israel~\cite{Israel2015EigenvalueDisTFLocalization} provided a non-asymptotic bound on the number of eigenvalues in the transition region. Fix $\eta\in(0,1/2]$. Given $\epsilon\in(0,1/2)$ and $T\geq 2$, then~\cite{Israel2015EigenvalueDisTFLocalization}
\begin{align}
\calN(\calO_{\setA_T,\setB};(\epsilon,1-\epsilon))\leq 2 C_\eta\left(\log\left(\frac{\log T}{\epsilon}\right)\right)^{1+\eta}\log\left(\frac{T}{\epsilon}\right),
\label{eq:nonasymptotic PSWF}\end{align}
where $C_\eta$ is a constant dependent on $\eta\in(0,\frac{1}{2}]$.

\item The earliest result on the eigenvalue distribution of the discrete time-frequency localization operator (where $\setG = \Z$, $\widehat \setG = [-\frac{1}{2},\frac{1}{2})$, $\setA_N = \left\{0,1,\ldots,N-1\right\}$  and $\setB = [-W,W]$ with $W\in(0,\frac{1}{2})$) comes from Slepian~\cite{Slepian78DPSS}, who showed that for any $b\in\R$, asymptotically the DPSS eigenvalue $\lambda_\ell (\calO(\setA_N,\setB))\rightarrow \frac{1}{1+e^{\pi b}}$ as $N\rightarrow \infty$ if $\ell =\lfloor 2NW + \frac{b}{\pi}\log N\rfloor$. This implies the asymptotic result: $
\calN(\calO_{\setA_N,\setB};(\epsilon,1-\epsilon))\sim \frac{2}{\pi^2}\log N\log\left(\frac{1}{\epsilon} -1\right).$
Recently, by examining the difference between the operator $\calO_{\setA_N,\setB}$ and the one formed by a partial DFT matrix, we have shown~\cite{KarnikZWRD_Fast,Karnik2016FAST} the following nonasymptotic result characterizing the $O(\log N\log{\frac{1}{\epsilon}})$ dependence:
\begin{align}
\calN(\calO_{\setA_N,\setB};(\epsilon,1-\epsilon)) \leq \left(\frac{8}{\pi^2} \log(8N) + 12\right) \log \left( \frac{15}{\eps} \right).
\label{eq:nonasymptotic DPSS}\end{align}
The right hand side is further improved to $\frac{2}{\pi^2}\log(4N)\left( \frac{4}{\eps(1-\eps)} \right)$ in \cite{karnik2020improved}.
\item We~\cite{zhu2017DFT} have also provided similar results for the eigenvalue distribution of discrete periodic time-frequency localization operator with sets $\setA_M$ and $\setB$ defined in \eqref{eq:set for DFT limiting}:
 \begin{align}
 	\calN(\calO_{\setA_M,\setB};(\epsilon,1-\epsilon)) \leq \left(\frac{8}{\pi^2} \log(8N) + 12\right) \log \left( \frac{15}{\eps} \right) + 4\max\Big(\frac{-\log \left( \frac{\pi}{32}\left(\left(\frac{M}{N}\right)^2-1\right) \epsilon \right)}{\log\left(\frac{M}{N}\right)},0\Big).
 \label{eq:nonasymptotic DPSS-v2}\end{align}
\end{itemize}
\end{remark}
\begin{remark}
It is also of particular interest to have a finer result on the number of eigenvalues that is greater than $\frac{1}{2}$ since this together with the size of the transition region gives us a complete understanding of the eigenvalue distribution.
\begin{itemize}[leftmargin=*]
\item Landau~\cite{landau1993density} establishes the number of PSWF eigenvalues that are greater than $\frac{1}{2}$ as follows
\begin{align}
\lambda_{(\lfloor T \rfloor-1)}(\calO_{\setA_T,\setB})\geq \frac{1}{2}\geq \lambda_{(\lceil T\rceil)}(\calO_{\setA_T,\setB}).
\label{eq:eig near 1/2}\end{align}
\item We~\cite{zhu2015approximating} provided a similar result for the DPSS eigenvalues.
\end{itemize}
\end{remark}

In the following two subsections, we review some applications of Theorem~\ref{thm:limiting operator}.


\subsection{Application: Communications}
\label{sec:tfapplicationscomm}

In \cite{franceschetti2015landau}, Franceschetti extended Landau's theorem~\cite{landau1975szego} for simple time and frequency intervals to other time and frequency sets of complicated shapes. Lim and Franceschetti~\cite{lim2017information} related  the number of degrees of freedom of the space of band-limited signals to the deterministic notions of capacity and entropy. Now we apply Theorem~\ref{thm:limiting operator} to the effective dimensionality of the ``band-limited signals'' observed over a finite set $\setA$ by utilizing the result in Section~\ref{sec:number of degrees}. To that end, we plug $\widehat \phi(\xi) = 1_{\setB}(\xi) = \left\{\begin{matrix} 1,& \xi \in \setB\\ 0, & \xi\notin \setB\end{matrix}\right.$, the indicator function on $\setB$, into \eqref{eq:set W} and get the following set of band-limited functions observed only over $\setA$:
\begin{align*}
\setW(\setA,1_{\setB}(\xi)) := \left\{x\in L_2(\setA): x(g) = \int_{\setB} \alpha(\xi)\chi_\xi(g) \dif \xi, \int \left|\alpha(\xi)\right|^2\leq 1, g\in \setA \right\}.
\end{align*}
When $\setA\subset \R^2$ represents a subset of time and space,
the number of degrees of freedom in the set $\setW(\setA,1_{\setB}(\xi))$ determines the total amount of information that can be transmitted in time and space by multiple-scattered electromagnetic waves~\cite{franceschetti2015landau}. Now we turn to compute the effective dimensionality of the general set $\setW(\setA,1_{\setB}(\xi))$. In this case, $|\widehat \phi (\xi)|^2 = 1_\setB(\xi)$ and the corresponding operator $\calA\calA^*$ defined in~\eqref{eq:AAadjoint} is equivalent to the time-frequency limiting operator $\calO_{\setA,\setB}$ in \eqref{eq:generalized time-band limiting operator}. Now Proposition~\ref{prop: n-width} implies that the effective dimensionality
$\calN(\setW(\setA,\widehat\phi(\xi)),\epsilon)$ is equal to the number of eigenvalues of $\calO_{\setA,\setB}$ that are greater than $\epsilon$,  which is given by Theorem~\ref{thm:limiting operator}. In words, the effective dimensionality of the set $\setW(\setA,1_{\setB}(\xi))$ is essentially $|\setA||\setB|$, and is insensitive to the level $\epsilon$ (as illustrated in \eqref{eq:nonasymptotic PSWF}-\eqref{eq:nonasymptotic DPSS-v2}, in many cases, this dimensionality only has $\log(\frac{1}{\epsilon})$ dependence on $\epsilon$).

\subsection{Application: Signal representation}
\label{sec:tfapplicationssignalrep}

In addition to the eigenvalues of the time-frequency limiting operator $\calT_{\setA}\calB_\setB\calT_{\setA}$, the eigenfunctions of  $\calT_{\setA}\calB_\setB\calT_{\setA}$ are also of significant importance, owing to their concentration in the time and frequency domains. To see this, let $u_\ell(g)$ be the $\ell$-th eigenfunction of $\calT_{\setA}\calB_\setB\calT_{\setA}$, corresponding to the $\ell$-th eigenvalue $\lambda_\ell (\calT_{\setA}\calB_\setB\calT_{\setA})$. Denoting the Fourier transform of $u_\ell(g)$ by $\widehat u_\ell(\xi)$,
we have
\begin{equation}\begin{split}
\int_{\setB}|\widehat u_\ell(\xi)|^2\dif \xi &= \left\langle \calT_\setB \calF u_\ell, \calT_\setB \calF u_\ell \right\rangle = \left\langle \calF^{-1}\calT_\setB \calF u_\ell,  u_\ell \right\rangle = \left\langle \calF^{-1}\calT_\setB \calF \calT_\setA u_\ell,  \calT_\setA u_\ell \right\rangle\\
& = \left\langle \calT_\setA\calF^{-1}\calT_\setB \calF \calT_\setA u_\ell,  \calT_\setA u_\ell \right\rangle = \lambda_\ell(\calT_\setA\calB_\setB\calT_\setA)\|\calT_\setA u_\ell\|^2,
\end{split}\label{eq:concen freq}\end{equation}
where the third equality follows because $u_\ell(g)$ is a time-limited signal (i.e., $\calT_\setA (u_\ell) = u_\ell$), and the last equality utilizes $\calT_\setA\calF^{-1}\calT_\setB \calF \calT_\setA u_\ell = \calT_\setA\calB_\setB\calT_\setA u_\ell = \lambda_\ell (\calT_{\setA}\calB_\setB\calT_{\setA})u_\ell$. In words, \eqref{eq:concen freq} states that the eigenfunctions $u_\ell$ have a proportion $\lambda_\ell(\calT_\setA\calB_\setB\calA_\setA)$ of energy within the band $\setB$, implying that even though the eigenfunctions are not exactly band-limited, their Fourier transform is mostly concentrated in the band $\setB$ when $\lambda_\ell(\calT_\setA\calB_\setB\calA_\setA)$ is close to 1. Thus, the first $\approx |\setA||\setB|$ eigenfunctions can be utilized as window functions for spectral estimation, and as a highly efficient basis for representing band-limited signals that are observed over a finite set $\setA$.

Recall that $\setW(\setA,1_{\setB}(\xi))$ (defined in~\eqref{eq:set W}) consists of band-limited signals observed over a finite set $\setA$. Applying Proposition~\ref{prop: n-width}, we compute the $n$-width of the set $\setW(\setA,1_{\setB}(\xi))$ as follows:
\[
d_n(\setW(\setA,1_{\setB}(\xi))) = \sqrt{\lambda_n(\calA\calA^*)} = \sqrt{\lambda_n(\calO_{\setA,\setB})}.
\]
By the definition of \eqref{eq:def n-width}, we know for any $x(g)\in \setW(\setA,1_{\setB}(\xi))$,
\[
\int_{\setA} \left| x(g)  - (\P_{\setU_n}x)(g) \right|^2 \dif g \leq \sqrt{\lambda_n(\calO_{\setA,\setB})},
\]
where $\setU_n$ is the subspace spanned by the first $n$ eigenvectors of $\calO_{\setA_\tau,\setB}$, i.e.,
\begin{align}
\setU_n := \text{span}\{u_0(g),u_1(g),\ldots,u_{n-1}(g)\}.
\label{eq:subspace U}\end{align}
Now we utilize Theorem~\ref{thm:limiting operator} to conclude that the representation residual $\sqrt{\lambda_n(\calO_{\setA,\setB})}$ is very small when $n$ is chosen slightly larger than $|\setA||\setB|$.

We now investigate the basis $\setU_n$ for representing time-limited version of characters $\chi_\xi(g)$ and band-limited signals.

\subsubsection{Approximation quality for time-limited  characters $\chi_\xi(g)$ }

We first restrict our focus to the simplest possible ``band-limited signals'' that are observed over a finite period: pure characters $\chi_\xi(g)$ when $g$ is limited to $\setA$. Without knowing the exact frequency $\xi$ in advance, we attempt to find an efficient low-dimensional basis for capturing the energy in any signal $\chi_\xi(g)$. To that end, we let $\setM_n\revise{\subset} L_2(\setA)$ denote an $n$-dimensional subspace of $L_2(\setA)$. We would like to minimize
\begin{align}
\int_{\setB} \|\chi_\xi - \P_{\setM_n}\chi_\xi\|_{L_2(\setA)}^2 \dif \xi.
\label{eq:best subsapce}\end{align}
The following result establishes the degree of approximation accuracy in a mean-squared error (MSE) sense provided by the subspace $\setU_n$ for representing the ``time-limited'' version of characters $\chi_\xi(g)$ (where $g$ is limited to $\setA_\tau$).

\begin{thm}\label{thm:subspace for characters}
For any $n\in \Z^+$, the optimal $n$-dimensional subspace which minimizes~\eqref{eq:best subsapce} is $\setU_n$. Furthermore, with this choice of subspace, we have
\begin{align*}
\frac{1}{|\setB|}\int_{\setB} \frac{\|\chi_\xi - \P_{\setU_n}\chi_\xi\|_{L_2(\setA)}^2}{\|\chi_\xi\|_{L_2(\setA)}^2} \dif \xi  = 1 -\frac{ \sum_{\ell=0}^{n-1}\lambda_{\ell}(\calO_{\setA,\setB})}{|\setA||\setB|},
\end{align*}
\revise{where $|\cdot|$ denotes the Haar measure.}
\end{thm}
The proof of Theorem~\ref{thm:subspace for characters} is given in Appendix~\ref{sec:prf subspace for characters}. Combined with Theorem~\ref{thm:limiting operator}, Theorem~\ref{thm:subspace for characters} implies that by choosing $n\approx |\setA||\setB|$, on average the subspace spanned by the first $n$ eigenfunctions of $\calT_{\setA}\calB_\setB\calT_{\setA}$ is expected to accurately represent time-limited characters within the band of interest. We note that the representation guarantee for time-limited characters $\{\calT_{\setA}\chi_\xi,\xi \in \setB\}$ can also be used for most band-limited signals that are observed over a finite set $\setA$. To see this, suppose $x(g)$ is a band-limited function which can be represented as
\[
x(g) = \int_\setB \widehat x(\xi) \chi_\xi(g)\dif \xi.
\]
An immediate consequence of the above equation is that one can view $\{\calT_{\setA}\chi_\xi,\xi \in \setB\}$ as the atoms for building $\calT_{\setA} x$:
\[
\calT_{\setA} x = \int_\setB \widehat x(\xi) \calT_{\setA}\chi_\xi\dif \xi.
\]

\subsubsection{Approximation quality for random band-limited signals}
We can also approach the representation ability of the subspace $\setU_n$ (defined in \eqref{eq:subspace U}) from a probabilistic perspective \revise{by generalizing \cite[Theorem 4.1]{DavenportWakin2012CSDPSS}.}
\begin{thm}\label{thm:random approach}
Let $x(g) = \chi_\xi(g), g\in \setA_\tau$ be a random function where $\xi$ is a random variable with uniform distribution on $\setB$.
Then we have
\[
\frac{\E\left[\|x - \P_{\setU_n}x\|_{L_2(\setA)}^2 \right]}{\E\left[\|x \|_{L_2(\setA)}^2 \right]} = 1 - \frac{\sum_{\ell=0}^{n-1}\lambda_{\ell}(\calO_{\setA,\setB})}{|\setA||\setB|}.
\]
\end{thm}

The proof of Theorem~\ref{thm:random approach} is given in Appendix \ref{sec:prf random approach}. With this result, we show that in a certain sense, most band-limited signals, when time-limited, are well-approximated by a signal within the subspace $\setU_n$. In particular, the following result  \revise{which generalizes\cite[Theorem 4.1]{DavenportWakin2012CSDPSS}} establishes that band-limited random processes, when time-limited, are in expectation well-approximated.

\begin{cor}\label{cor:random process}
Let $x(g),g\in \setG$ be a zero-mean wide sense stationary random process over the group $\setG$ with power spectrum
\[
P_x(\xi) = \left\{\begin{matrix} \frac{1}{|\setB|}, & \xi \in \setB,\\ 0,& \text{otherwise}.\end{matrix}\right.
\]
Suppose we only observe $x$ over the set $\setA_\tau$. Then we have
\[
\frac{\E\left[\|x - \P_{\setU_n}x\|_{L_2(\setA)}^2 \right]}{\E\left[\|x \|_{L_2(\setA)}^2 \right]} = 1 - \frac{\sum_{\ell=0}^{n-1}\lambda_{\ell}(\calO_{\setA,\setB})}{|\setA||\setB|}.
\]
\end{cor}

As in our discussion following Theorem~\ref{thm:limiting operator}, the term $1 - \frac{\sum_{\ell=0}^{n-1}\lambda_{\ell}(\calO_{\setA,\setB})}{|\setA||\setB|}$ appearing in Theorem \ref{thm:random approach} and Corollary~\ref{cor:random process} can be very small when we choose $n$ slightly larger than $|\setA||\setB|$. This suggests that in a probabilistic sense, most band-limited functions, when time-limited, will be well-approximated by a small number of eigenfunctions of the operator $\calO_{\setA,\setB}$.


\subsection{Applications in the common time and frequency domains}
\label{sec:tfappreview}

We now review several applications involving the time-frequency limiting operator $\calO_{\setA,\setB}$ in the common time and frequency domains, where the eigenfunctions correspond to DPSSs, PSWFs, and PDPSSs.

It follows from \eqref{eq:concen freq} that, among all the functions that are time-limited to the set $\setA$, the first eigenfunction $u_0(g)$ is maximally concentrated in the subset $\setB$ of the frequency domain. Motivated by this result, the first DPSS vector is utilized as a filter for super-resolution~\cite{eftekhari2015greed}.
In \cite{thomson1982spectrum}, the first $\approx 2NW $ DPSS vectors are utilized as window functions (a.k.a.\ tapers) for spectral estimation. The multitaper method~\cite{thomson1982spectrum} averages the tapered estimates with the DPSS vectors, and has been used in a variety of scientific applications including  statistical signal analysis~\cite{cox1996spectral}, geophysics and cosmology~\cite{dahlen2008spectral}.
By exploiting the fact that the number of DPSS eigenvalues in the transition region grows like $O(\log N\log{\frac{1}{\epsilon}})$ as in \eqref{eq:nonasymptotic DPSS}, the very recent work \cite{karnik2021thomson} provided nonasymptotic bounds on some statistical properties of the multitaper spectral estimate as well as a fast algorithm for evaluating the estimate.

By exploiting the concentration behavior of the PSWFs in the time and frequency domains (where $\setG = \R$ and $\widehat \setG = \R$), Xiao et al.~\cite{xiao2001prolate} utilized the PSWFs to  numerically construct
quadratures, interpolation and differentiation formulae for band-limited
functions. Gosse~\cite{gosse2013compressed} constructed a PSWF dictionary consisting of the first few PSWFs for recovering smooth functions from random samples.  The connection between time-frequency localization of multiband signals and sampling theory for such signals was investigated in~\cite{Izu2009TimeFrequencyLocalization}. In~\cite{senay2009reconstruction,senay2008compressive}, the authors also considered a PSWF dictionary for reconstruction of electroencephalography (EEG) signals and time-limited signals that are also nearly band-limited from nonuniform samples. Chen and Vaidyanathan~\cite{chen2008mimo} utilized the PSWFs to represent the clutter subspace (and hence mitigate the clutter), facilitating space-time adaptive processing for multiple-input multiple-output (MIMO) radar systems; see also~\cite{yang2013robust,du2016robust}.

DPSSs, the discrete counterpart of PSWFs, also have proved to be useful in numerous signal processing applications since they provide a highly efficient basis for representing sampled band-limited signals. DPSSs can be utilized to find the minimum energy, infinite-length band-limited sequence that extrapolates a given finite vector of samples~\cite{Slepian78DPSS}. In~\cite{zemen2005channelEstim,zemen2007minimum}, Zemen et al. expressed the time-varying subcarrier coefficients in a DPSS basis for estimating time-varying channels in wireless communication systems. A similar idea is also utilized for channel estimation  in  Orthogonal Frequency Division Multiplexing (OFDM) systems~\cite{sreepada2016channel}, for receiver antenna selection~\cite{saleh2012receive}, etc. The modulated DPSSs can also be useful   for mitigating wall clutter and detecting targets behind the wall in through-the-wall radar imaging~\cite{Zhu2015targetDetectDPSS,zhu2016dimensionality},  and for interference cancellation in a wideband compressive radio receiver (WCRR) architecture~\cite{davenport2010wideband}. The performance (such as the detection probability) of the DPSS basis (and other similar bases corresponding to the time-frequency limiting operator $\calO_{\setA,\setB}$) for identifying unresolved targets was recently analyzed in\cite{bosse2018subspace}.
By modulating the baseband DPSS vectors to different frequency bands and then concatenating these dictionaries, one can construct a new dictionary that provides an efficient representation of sampled multiband signals~\cite{{DavenportWakin2012CSDPSS},zhu2015approximating}. Sejdi\'{c} et al.~\cite{SejdicICASSP2008ChannelEstimationDPSS} proposed one such dictionary to provide a sparse representation for fading channels and improve channel estimation accuracy. The multiband modulated DPSS dictionaries have been utilized for the recovery of sampled multiband signals from random measurements~\cite{DavenportWakin2012CSDPSS}, and for the recovery of physiological signals from compressive measurements~\cite{sejdic2012compressive}. Such dictionaries are also utilized for cancelling wall clutter~\cite{AhmadQianAmin2015WallCluterDPSS}.

The periodic DPSSs (PDPSSs, where $\setG = \Z_N$ and $\widehat\setG = \Z_N$) are the finite-length vectors whose discrete Fourier transform (DFT) is most concentrated in a given bandwidth (as appearing in \eqref{eq:set for DFT limiting}). The PDPSSs have been utilized for extrapolation and spectral estimation of periodic discrete-time signals~\cite{jain1981extrapolation}, for limited-angle reconstruction in tomography~\cite{grunbaum1981eigenvectors }, for Fourier extension~\cite{matthysen2016fast}, and in~\cite{hogan2015wavelet}, the bandpass PDPSSs were used as a numerical approximation to the bandpass PSWFs for studying synchrony in sampled EEG signals.

Finally, the eigenvalue concentration behavior in Theorem~\ref{thm:limiting operator} can also be exploited for solving a linear system involving the Toeplitz operator $\calO_{\setA,\setB}$: $y = \calO_{\setA,\setB} x$. Since the operator $\calO_{\setA,\setB}$ has a mass of eigenvalues that are very close to $0$, the system is often solved by using the rank-$K$ pseudoinverse of $\calO_{\setA,\setB}$ where $K\approx |\setA||\setB|$. In the case where the Toeplitz operator is the prolate matrix $\mB_{N,W}$ defined in \eqref{eq:prolate matrix B_NW}, its truncated pseudoinverse is well approximated as the sum of $\mB_{N,W}^*$ (which is equal to $\mB_{N,W}$) and a low-rank matrix~\cite{matthysen2016fast,Karnik2016FAST} since most of the eigenvalues of $\mB_{N,W}^*$ are very close to either  $1$ or $0$. By utilizing the fact that $\mB_{N,W}$ is a Toeplitz matrix and $\mB_{N,W}\vx$ has a fast implementation via the FFT, an efficient method for solving the system $\vy = \mB_{N,W}\vx$ can be developed; such a method has been utilized for linear prediction of band-limited signals based on past samples and the Fourier extension~\cite{matthysen2016fast,Karnik2016FAST}.

\subsection{Questions}
Inspired by the applications listed above, we raise several questions concerning Theorems~\ref{thm:limiting operator} and \ref{thm:subspace for characters}. Following from the two remarks after Theorem \ref{thm:limiting operator}, two natural questions are:
\begin{question}
Can we improve the second term in \eqref{eq:sum of squares}? Furthermore, what nonasymptotic result (like \eqref{eq:nonasymptotic PSWF} for the PSWF eigenvalues and~\eqref{eq:nonasymptotic DPSS} for the DPSS eigenvalues) can we obtain for the number of eigenvalues of the Toeplitz operator $\calO_{\setA,\setB}$ within the transition region $(\epsilon,1-\epsilon)$?
\end{question}

\begin{question}
Can we extend \eqref{eq:eig near 1/2} to the general time-frequency limiting operator $\calO_{\setA,\setB}$?
\end{question}

Another related important question concerns how accurately the subspace spanned by the first $n$ eigenfunctions of $\calT_{\setA}\calB_\setB\calT_{\setA}$ can represent each individual time-limited character $\calT_\setA\chi_\xi$
with $\xi\in\setB$. Theorem \ref{thm:subspace for characters} ensures that  accuracy is guaranteed in the MSE sense if one chooses $n\approx |\setA||\setB|$ such that the sum of the remaining eigenvalues of $\calT_{\setA}\calB_\setB\calT_{\setA}$ is small.
We suspect that a uniform guarantee for each $\calT_\setA\chi_\xi$ can also be obtained since the derivative of $\|\chi_\xi\|_{L_2(\setA)}^2$ is bounded, given a finer result concerning the eigenvalue distribution for $\calT_{\setA}\calB_\setB\calT_{\setA}$.
Using the approach utilized in~\cite{zhu2017fast,zhu2018roast} with a theorem of Bernstein for trigonometric polynomials~\cite{schaeffer1941inequalities}, we can have an approximation guarantee for the DPSS basis in representing each complex exponential $\ve_f:=\begin{bmatrix}e^{j2\pi f 0} & \cdots & e^{j2\pi f(N-1)}\end{bmatrix}^{\T}$ with frequency $f$ inside a band of interest; this provides a non-asymptotic guarantee which improves upon our previous work \cite{zhu2015approximating}.

\begin{thm}\label{thm:fast for sinusoid}(Representation guarantee for pure sinusoids with DPSSs)
Let $N \in \N$ and $W \in (0, \tfrac{1}{2})$ be given. Also let $[\mS]_K$ be an $N\times K$ matrix consisting of  the first $K$ DPSS vectors. Then for any $\eps \in (0,\tfrac{1}{2})$, the orthobasis $[\mS]_K$ satisfies
\[
\frac{\left\|\ve_f - [\mS]_K[\mS]_K^*\ve_f\right\|_2^2}{\|\ve_f\|_2^2} \leq \epsilon
\]
for all $f\in[-W,W]$ with
\[
K = 2NW + O\left(\log (N)\log\left(\frac{1}{\epsilon^2}\right)\right).\]
\end{thm}

\begin{question}
More generally, what uniform guarantee can we have for each time-limited character $\calT_\setA\chi_\xi$ in the subspace spanned by the first $n$ eigenfunctions of $\calT_{\setA}\calB_\setB\calT_{\setA}$?
\end{question}

\section*{Acknowledgements}

This work was supported by NSF grants CCF-1409261, CCF-1704204, and CCF-2008460. We gratefully acknowledge Justin Romberg for pointing out a connection between subspace approximations and Toeplitz operators, and Mark Davenport and Santhosh Karnik for valuable discussions. 
\appendix

\section{Proof of Proposition~\ref{prop: n-width}}\label{sec:prf prop n-width}

\begin{proof}[Proof of Proposition~\ref{prop: n-width}]
We have
\begin{align*}
d_n(\setW(\setA,\widehat\phi(\xi))) &= \inf_{\setM_n}\sup_{x\in\setW(\setA,\widehat\phi(\xi))}\inf_{y \in \setM_n}\|x - y\|_{L_2(\setA)}\\
& = \inf_{\setM_n}\sup_{\|\alpha\|\leq 1}\left\|\calA\alpha - \P_{\setM_n}\right\|\\
& = \inf_{\setM_n}\sup_{z\perp\setM_n}\sup_{\|\alpha\|\leq 1}\frac{\left| \left\langle \calA\alpha,z\right\rangle\right|}{\|z\|}\\
& = \inf_{\setM_n}\sup_{z\perp\setM_n}\sup_{\|\alpha\|\leq 1}\frac{\left| \left\langle \alpha,\calA^* z\right\rangle\right|}{\|z\|}\\
& = \inf_{\setM_n}\sup_{z\perp\setM_n}\frac{\left\|\calA^* z\right\|}{\|z\|}\\
& = \inf_{\setM_n}\sup_{z\perp\setM_n}\frac{\sqrt{\left\langle\calA\calA^* z,z\right\rangle}}{\|z\|}\\
& = \sqrt{\lambda_n},
\end{align*}
where the last line follows from the Weyl-Courant minimax theorem.
\end{proof}

\section{Proof of Theorem~\ref{thm:limiting operator}}\label{sec:prf thm limiting operator}
\begin{proof}[Proof of Theorem~\ref{thm:limiting operator}]
We first note that $\chi_\xi(0) = 1$ for all $\xi\in\widehat\setG$. Thus, we have
\begin{align}
\sum_{\ell} \lambda_{\ell}(\calO_{{\setA_\tau},\setB}) = \int_{{\setA_\tau}} K_\setB(0)\dif h = |{\setA_\tau}|\int_\setB \chi_\xi(0)\dif \xi = |{\setA_\tau}| |\setB|,
\label{eq:sum = AB}\end{align}
\revise{where the first equality follows because $\calO_{{\setA_\tau},\setB}$ is a trace class operator \cite{deitmar2015trace}.}
We write the operator $(\calT_{\setA_\tau}\calB_\setB\calT_{\setA_\tau})^2$ as
\begin{align*}
(\calT_{\setA_\tau}\calB_\setB\calT_{\setA_\tau}\calB_\setB\calT_{\setA_\tau} x)(g) &= \int_{\setA_\tau} K_\setB(\widetilde h^{-1}\circ g) \left(\int_{\setA_\tau} K_\setB(h^{-1}\circ\widetilde h)   x(h) \dif h \right) \dif \widetilde h\\
& = \int_{\setA_\tau}  \left(\int_{\setA_\tau} K_\setB(\widetilde h^{-1}\circ g) K_\setB(h^{-1}\circ\widetilde h)  \dif \widetilde h \right) x(h) \dif h.
\end{align*}
Thus,
\begin{align*}
\sum_\ell\lambda^2_{\ell}(\calO_{{\setA_\tau},\setB})
& = \int_{\setA_\tau}  \int_{\setA_\tau} K_\setB(\widetilde h^{-1}\circ h) K_\setB(h^{-1}\circ \widetilde h)  \dif \widetilde h \dif h = \int_{\setA_\tau}  \int_{\setA_\tau} \left|K_\setB(h^{-1}\circ\widetilde h)\right|^2   \dif \widetilde h \dif h,
\end{align*}
where we use the fact that $K_\setB(h^{-1}\circ g) = \int_\setB \chi_\xi(h^{-1}\circ g) \dif \xi = (\int_\setB \chi_\xi(g^{-1}\circ h) \dif \xi)^*$ since $\chi_\xi(-g) = \chi_\xi^*(g)$. Applying the change of variable $\widetilde h = h\circ\overline h$, we obtain
\begin{align}
\sum_\ell\lambda^2_{\ell}(\calO_{{\setA_\tau},\setB}) = \int_{\setA_\tau}  \int_{{\setA_\tau}-h} \left|K_\setB(\overline h)\right|^2   \dif \overline h \dif h =  \int_{\setA_\tau}  \kappa_{{\setA_\tau},\setB}(h) \dif h,
\label{eq:sum of squares KB}\end{align}
where $\kappa_{\set{\setA_\tau},\setB}(h) = \int_{{\setA_\tau}-h}\left|K_\setB(\overline h)\right|^2   \dif \overline h\geq 0$. The function $\kappa_{\set{\setA_\tau},\setB}(h)$ is dominated as
\begin{align*}
\kappa_{\set{\setA_\tau},\setB}(h) \leq &  \int_{\setG}\left|K_\setB(\overline h)\right|^2   \dif \overline h = \int_{\setG}\left|\int_\setB \chi_\xi(\overline h) \dif \xi\right|^2 \dif \overline h = \int_{\widehat\setG} \revise{1_{\xi \in \setB}} \dif \xi =  |\setB|,
\end{align*}
where we use Parseval's theorem \revise{by viewing $\int_\setB \chi_\xi(\overline h) \dif \xi$ as the inverse Fourier transform of a window function supported on $\setB$, i.e. $1_{\xi \in \setB} = \begin{cases}1, &  \xi \in \setB,\\ 0, & \text{otherwise}, \end{cases}$.}

On the other hand, we have $
\lim_{\tau\rightarrow \infty}\kappa_{\set\setA_\tau,\setB}(h) = \int_{\setG}\left|\int_\setB \chi_\xi(\overline h) \dif \xi\right|^2 \dif \overline h  = |\setB|$
for all $h\in\setG$.
It follows that $
\lim_{\tau\rightarrow \infty}\sum_\ell\lambda_{\ell}^2(\calO_{\setA_\tau,\setB}) = \int_{\setA_\tau} |\setB| \dif h = |\setA_\tau||\setB|.$
Thus, we have
\begin{align}
\sum_\ell\lambda^2_{\ell}(\calO_{\setA_\tau,\setB}) = |\setA_\tau||\setB| - o(|\setA_\tau||\setB|).
\label{eq:sum of squars = AB - o}\end{align}
Subtracting~\eqref{eq:sum of squars = AB - o} from \eqref{eq:sum = AB} gives
\begin{align}
\sum_{\ell}  \lambda_{\ell}(\calO_{\setA_\tau,\setB}) \left(1 - \lambda_{\ell}(\calO_{\setA_\tau,\setB})\right) =  o(|\setA_\tau||\setB|).
\label{eq:lambda 1-lambda}\end{align}
Utilizing the fact that $0\leq \lambda_{\ell}(\calO_{\setA_\tau,\setB}) \leq 1$, we have
\[
\revise{\epsilon (1-\epsilon)} \calN(\calO_{\setA_\tau,\setB};(\epsilon,1-\epsilon)) \leq \sum_{\ell}  \lambda_{\ell}(\calO_{\setA_\tau,\setB}) \left(1 - \lambda_{\ell}(\calO_{\setA_\tau,\setB})\right) =  o(|\setA_\tau||\setB|).
\]
On the other hand, \eqref{eq:lambda 1-lambda} also implies that
\begin{align}
\sum_{\ell:\lambda_{\ell}(\calO_{\setA_\tau,\setB})< 1-\epsilon}\epsilon \lambda_{\ell}(\calO_{\setA_\tau,\setB})< \sum_{\ell}  \lambda_{\ell}(\calO_{\setA_\tau,\setB}) \left(1 - \lambda_{\ell}(\calO_{\setA_\tau,\setB})\right) =  o(|\setA_\tau||\setB|).
\label{eq:control tail}\end{align}
Plugging this term into \eqref{eq:sum = AB} gives
\begin{align*}
|\setA_\tau||\setB| = \sum_{\ell} \lambda_{\ell}(\calO_{\setA_\tau,\setB}) &= \sum_{\ell:\lambda_{\ell}(\calO_{\setA_\tau,\setB})\geq 1-\epsilon} \lambda_{\ell}(\calO_{\setA_\tau,\setB}) + \sum_{\ell:\lambda_{\ell}(\calO_{\setA_\tau,\setB}) \revise{<}1-\epsilon}\lambda_{\ell}(\calO_{\setA_\tau,\setB})\\
& = \sum_{\ell:\lambda_{\ell}(\calO_{\setA_\tau,\setB})\geq 1-\epsilon} \lambda_{\ell}(\calO_{\setA_\tau,\setB}) + o(|\setA_\tau||\setB|).
\end{align*}
Similarly, plugging~\eqref{eq:control tail} into \eqref{eq:sum of squars = AB - o} gives
\begin{align*}
|\setA_\tau||\setB| =  \sum_{\ell:\lambda_{\ell}(\calO_{\setA_\tau,\setB})\geq 1-\epsilon} \lambda^2_{\ell}(\calO_{\setA_\tau,\setB}) + o(|\setA_\tau||\setB|).
\end{align*}
Combining the above two equations and the fact that $\lambda_{\ell}(\calO_{\setA_\tau,\setB})\leq 1$, we have
\begin{align}
\sum_{\ell:\lambda_{\ell}(\calO_{\setA_\tau,\setB})\geq 1-\epsilon} \lambda_{\ell}(\calO_{\setA_\tau,\setB}) -  \lambda^2_{\ell}(\calO_{\setA_\tau,\setB}) = o(|\setA_\tau||\setB|).
\label{eq:large eigen sum - sum squares}\end{align}
On one hand, combining \eqref{eq:large eigen sum - sum squares} with
\[
\sum_{\ell:\lambda_{\ell}(\calO_{\setA_\tau,\setB})\geq 1-\epsilon} \lambda_{\ell}(\calO_{\setA_\tau,\setB}) -  \lambda^2_{\ell}(\calO_{\setA_\tau,\setB}) \leq \sum_{\ell:\lambda_{\ell}(\calO_{\setA_\tau,\setB})\geq 1-\epsilon} 1 -  \lambda_{\ell}(\calO_{\setA_\tau,\setB})
\]
gives
\begin{align*}
\calN(\calO_{\setA_\tau,\setB};[1-\epsilon,1]) - \sum_{\ell:\lambda_{\ell}(\calO_{\setA_\tau,\setB})\geq 1-\epsilon} \lambda_{\ell}(\calO_{\setA_\tau,\setB}) = \sum_{\ell:\lambda_{\ell}(\calO_{\setA_\tau,\setB})\geq 1-\epsilon} 1 -  \lambda_{\ell}(\calO_{\setA_\tau,\setB}) \geq o(|\setA_\tau||\setB|),
\end{align*}
which further implies
\[
\calN(\calO_{\setA_\tau,\setB};[1-\epsilon,1]) \geq |\setA_\tau||\setB| - o(|\setA_\tau||\setB|).
\]
On the other hand, using \eqref{eq:large eigen sum - sum squares} and
\[
\sum_{\ell:\lambda_{\ell}(\calO_{\setA_\tau,\setB})\geq 1-\epsilon} \lambda_{\ell}(\calO_{\setA_\tau,\setB}) -  \lambda^2_{\ell}(\calO_{\setA_\tau,\setB}) \geq (1-\epsilon) \sum_{\ell:\lambda_{\ell}(\calO_{\setA_\tau,\setB})\leq 1-\epsilon} 1 -  \lambda_{\ell}(\calO_{\setA_\tau,\setB}), \]
we also have
\[
\calN(\calO_{\setA_\tau,\setB};[1-\epsilon,1]) - \sum_{\ell:\lambda_{\ell}(\calO_{\setA_\tau,\setB})\geq 1-\epsilon} \lambda_{\ell}(\calO_{\setA_\tau,\setB}) = \sum_{\ell:\lambda_{\ell}(\calO_{\setA_\tau,\setB})\geq 1-\epsilon} 1 -  \lambda_{\ell}(\calO_{\setA_\tau,\setB}) \leq o(|\setA_\tau||\setB|),
\]
which further implies $
\calN(\calO_{\setA_\tau,\setB};[1-\epsilon,1]) \leq |\setA_\tau||\setB| + o(|\setA_\tau||\setB|).$
Thus, $
\lim_{\tau\rightarrow \infty}\frac{\calN(\calO_{\setA_\tau,\setB};[1-\epsilon,1])}{|\setA_\tau|} = |\setB|.$

\end{proof}

\section{Proof of Theorem~\ref{thm:subspace for characters}}\label{sec:prf subspace for characters}
\begin{proof}[Proof of Theorem~\ref{thm:subspace for characters}]
We first recall the eigendecompostion of $\calO_{\setA_\tau,\setB} = \sum_{\ell\geq 0}\lambda_\ell u_\ell u_\ell^*$, where $\lambda_\ell$ is short for $\lambda_\ell(\calO_{\setA_\tau,\setB})$. Utilizing the fact that $u_\ell,\ell = 0,1,\ldots$ is a complete orthonormal basis for $L_2(\setA_\tau)$, we rewrite the function in \eqref{eq:best subsapce}:
\begin{align*}
\left\|\chi_\xi(g) - \P_{\setM_n}\chi_\xi(g)\right\|_{L_2(\setA_\tau)}^2&=
\sum_{\ell} \left|\left\langle(I-\P_{\setM_n})\chi_\xi(g),u_\ell(g)\right\rangle_{L_2(\setA_\tau)}\right|^2\\
& = \sum_{\ell} \left\langle \left\langle(I-\P_{\setM_n})\chi_\xi(g)\chi_\xi^*(h),u_\ell^*(h)\right\rangle_{L_2(\setA_\tau)},u_\ell(g)\right\rangle_{L_2(\setA_\tau)}\\
 &= \sum_{\ell} \left\langle \left\langle(I-\P_{\setM_n})\chi_\xi(h^{-1}\circ g),u_\ell^*(h)\right\rangle_{L_2(\setA_\tau)},u_\ell(g)\right\rangle_{L_2(\setA_\tau)}
\end{align*}
where the second equality utilized the fact that $\P_{\setM_n}$ is the orthogonal projector onto the subspace $\setM_n$, and
 $\sum_{\ell} \left\langle \left\langle(I-\P_{\setM_n})\chi_\xi(g)\chi_\xi^*(h),u_\ell^*(h)\right\rangle_{L_2(\setA_\tau)},u_\ell(h)\right\rangle_{L_2(\setA_\tau)}$ is equivalent to the trace of $ (I-\P_{\setM_n}) \chi_\xi (g) \chi_\xi^* (h)$.
Plugging this equation into \eqref{eq:best subsapce} gives
\begin{align*}
\int_{\setB} \|\chi_\xi(g) - \P_{\setM_n}\chi_\xi(g)\|_{L_2(\setA_\tau)}^2 \dif \xi  &= \int_{\setB}\sum_{\ell} \left\langle \left\langle(I-\P_{\setM_n})\chi_\xi(\theta^{-1}\circ g),u_\ell^*(h)\right\rangle_{L_2(\setA_\tau)},u_\ell(g)\right\rangle_{L_2(\setA_\tau)} \dif \xi\\
& = \sum_{\ell} \int_{\setB} \left\langle \left\langle(I-\P_{\setM_n})\chi_\xi(\theta^{-1}\circ g),u_\ell^*(h)\right\rangle_{L_2(\setA_\tau)},u_\ell(g)\right\rangle_{L_2(\setA_\tau)} \dif \xi\\
& =  \sum_{\ell}\left\langle(I-\P_{\setM_n})\calO_{\setA_\tau,\setB}u_\ell,u_\ell\right\rangle_{L_2(\setA_\tau)} = \sum_{\ell}\lambda_\ell\left\langle(I-\P_{\setM_n})u_\ell,u_\ell\right\rangle_{L_2(\setA_\tau)}
\end{align*}
where the second line follows from monotone convergence theorem (since each term inside the summation is nonnegative). Thus, we conclude that the optimal $n$-dimensional subspace which minimizes the last term in the above equation is $\setU_n$ (which is spanned by the first $n$ eigenfunctions). With this choice of subspace and \eqref{eq:sum of squares} that $\sum_{\ell} \lambda_\ell  = |\setA_\tau||\setB|$, we have
\begin{align*}
\int_{\setB} \|\chi_\xi(g) - \P_{\setU_n}\chi_\xi(g)\|_{L_2(\setA_\tau)}^2 \dif \xi & = \sum_{\ell\geq n} \lambda_\ell  = |\setA_\tau||\setB| - \sum_{\ell=0}^{n-1}\lambda_{\ell}
\end{align*}
The proof is completed by noting that $\|\chi_\xi(g) \|_{L_2(\setA_\tau)}^2 = |\setA_\tau|$ for any $\xi\in \setB$.
\end{proof}

\section{Proof of Theorem~\ref{thm:random approach}}\label{sec:prf random approach}
\begin{proof}[Proof of Theorem~\ref{thm:random approach}]
First let $\nu$ be a random variable with uniform distribution on $[0,2\pi)$. We define the random vector
\[
r(g) = r(g;\xi,\nu) = \chi_\xi(g)e^{j\nu},
\]
where the term $e^{j\nu}$ acts as a phase randomizer and ensures that $r$ is zero-mean:
\begin{align*}
\E\left[r(g)\right] = \frac{1}{|\setB|2\pi}\int_\setB \chi_\xi(g)e^{j\nu} \dif \xi \dif \nu = \frac{1}{|\setB|2\pi}\int_\setB\chi_\xi(g) \dif \xi \int_0^{2\pi} e^{j\nu} \dif \nu = 0
\end{align*}
for all $g\in \setA_\tau$.

Now we compute the autocorrelation $R$ of the random variable $r$ as
\begin{equation}\begin{split}
R(g,h) &= \E\left[r(g) r^*(h) \right] = \E\left[\left(\chi_\xi(g)e^{j\nu}\right) \left(\chi_\xi^*(h)e^{-j\nu}\right) \right] = \E\left[\chi_\xi(h^{-1}\circ g) \right]\\
& = \frac{1}{|\setB|}\int_\setB \chi_\xi(h^{-1}\circ g) \dif \xi = \frac{1}{|\setB|}K_\setB(h^{-1}\circ g)
\end{split}\label{eq:R}\end{equation}
for all $h,g\in \setA_\tau$. Here $K_\setB$ is defined in \eqref{eq:def KB}. Note that $K_\setB(h^{-1}\circ g)$ with $h,g\in \setA_\tau$ is the kernel of the Toeplitz operator $\calO_{\setA_\tau,\setB}$.
Now it follows from the Karhunen-Lo\`{e}ve (KL) transfrom~\cite{stark1986probability} that
\[
\E\left[\|r - \P_{\setU_n}r\|_{L_2(\setA_\tau)}^2 \right] = \frac{1}{|\setB|} \sum_{\ell\geq n}\lambda_{\ell}(\calO_{\setA_\tau,\setB})= |\setB| - \sum_{\ell=0}^{n-1}\lambda_{\ell}(\calO_{\setA_\tau,\setB}).
\]

We then compute the expectation for the energy of $r$ as
\[
\E\left[\|r \|_{L_2(\setA_\tau)}^2 \right] = \frac{1}{|\setB|}\frac{1}{2\pi}\int_\setB |\chi_\xi(g)e^{j\nu}|^2 \dif \xi \dif \nu = |\setA_\tau|.
\]
The proof is completed by noting that $\E\left[\|r - \P_{\setU_n}r\|_{L_2(\setA_\tau)}^2 \right] = \E\left[\|x - \P_{\setU_n}x\|_{L_2(\setA_\tau)}^2 \right]$ and $\E\left[\|r \|_{L_2(\setA_\tau)}^2 \right] = \E\left[\|x \|_{L_2(\setA_\tau)}^2 \right]$.
\end{proof}

\bibliographystyle{abbrv}
\bibliography{PhDThesis}

\end{document}